\pdfoutput=1
\pdfoutput=1
\pdfoutput=1
\pdfoutput=1
\pdfoutput=1
\documentclass[11pt]{article}
\usepackage{amssymb}
\usepackage{multirow}
\usepackage{amssymb}
\usepackage{amsmath}
\usepackage{amsthm}
\usepackage{array}
\usepackage{graphicx}
\usepackage{subfigure}
\usepackage{epstopdf}
\usepackage{cases}
\usepackage{color}
\usepackage{txfonts} % ĬÈÏ×ÖÌåtimes new roman
\usepackage{booktabs} %±í¸ñÏß¼Ó´Ö
\usepackage{float}
\usepackage{multirow}
\usepackage{picinpar}
\usepackage[colorlinks,
            linkcolor=red,
            anchorcolor=green,
            citecolor=blue
            ]{hyperref}
\usepackage{enumerate}
\usepackage{cases}
\usepackage[marginal]{footmisc}
\usepackage{fourier}
\usepackage{bookmark}
\usepackage[numbers,sort&compress]{natbib}
\begin{document}

\newcommand{\defi}{\stackrel{\Delta}{=}}
\newcommand{\A}{{\cal A}}
\newcommand{\B}{{\cal B}}
\newcommand{\U}{{\cal U}}
\newcommand{\G}{{\cal G}}
\newcommand{\cZ}{{\cal Z}}
\newcommand\one{\hbox{1\kern-2.4pt l }}
\newcommand{\Item}{\refstepcounter{Ictr}\item[\left(\theIctr\right)]}
\newcommand{\QQ}{\hphantom{MMMMMMM}}

\newtheorem{Theorem}{Theorem}[section]
\newtheorem{Lemma}{Lemma}[section]
\newtheorem{Corollary}{Corollary}[section]
\newtheorem{Remark}{Remark}[section]
\newtheorem{Example}{Example}[section]
\newtheorem{Proposition}{Proposition}[section]
\newtheorem{Property}{Property}[section]
\newtheorem{Assumption}{Assumption}[section]
\newtheorem{Definition}{Definition}[section]
\newtheorem{Construction}{Construction}[section]
\newtheorem{Condition}{Condition}[section]
\newtheorem{Exa}[Theorem]{Example}
\newcounter{claim_nb}[Theorem]
\setcounter{claim_nb}{0}
\newtheorem{claim}[claim_nb]{Claim}
\newenvironment{cproof}
{\begin{proof}
 [Proof.]
 \vspace{-3.2\parsep}}
{\renewcommand{\qed}{\hfill $\Diamond$} \end{proof}}
\newcommand{\erhao}{\fontsize{21pt}{\baselineskip}\selectfont}
\newcommand{\xiaoerhao}{\fontsize{18pt}{\baselineskip}\selectfont}
\newcommand{\sanhao}{\fontsize{15.75pt}{\baselineskip}\selectfont}
\newcommand{\sihao}{\fontsize{14pt}{\baselineskip}\selectfont}
\newcommand{\xiaosihao}{\fontsize{12pt}{\baselineskip}\selectfont}
\newcommand{\wuhao}{\fontsize{10.5pt}{\baselineskip}\selectfont}
\newcommand{\xiaowuhao}{\fontsize{9pt}{\baselineskip}\selectfont}
\newcommand{\liuhao}{\fontsize{7.875pt}{\baselineskip}\selectfont}
\newcommand{\qihao}{\fontsize{5.25pt}{\baselineskip}\selectfont}
\newcounter{Ictr}
\renewcommand{\theequation}{%\thesection.
\arabic{equation}}
\renewcommand{\thefootnote}{\fnsymbol{footnote}}

\def\A{\mathcal{A}}

\def\C{\mathcal{C}}

\def\V{\mathcal{V}}

\def\I{\mathcal{I}}

\def\Y{\mathcal{Y}}

\def\X{\mathcal{X}}

\def\J{\mathcal{J}}

\def\Q{\mathcal{Q}}

\def\W{\mathcal{W}}

\def\S{\mathcal{S}}

\def\T{\mathcal{T}}

\def\L{\mathcal{L}}

\def\M{\mathcal{M}}

\def\N{\mathcal{N}}
\def\R{\mathbb{R}}
\def\H{\mathbb{H}}

\title{}
\author{}
%%%%%%%%%%%%%%%%%%%%%%%%%%%%%%%%%%%%%%%%%%%%%%%%%%%%%%%%%%%%%%%%%%%%%%%%
\begin{center}
\topskip2cm
\LARGE{\bf Safe Feature Identification Rule for Fused Lasso by An Extra Dual Variable}
\end{center}

\begin{center}
\renewcommand{\thefootnote}{\fnsymbol{footnote}}Pan Shang,  Huangyue Chen, Lingchen Kong \footnote{Address: Pan Shang and Huangyue Chen are  with Academy of Mathematics and Systems Science, Chinese Academy of Sciences, Beijing 100190, China; Lingchen Kong is with the School of Mathematics and Statistics, Beijing Jiaotong University, Beijing, 100044, China.\\
E-mail: pshang@amss.ac.cn, hychen2022@amss.ac.cn, konglchen@126.com}\\
\today
%(June 14th, 2021)\\
\end{center}
\vskip4pt
\textbf{Abstract:} Fused Lasso was proposed to characterize the sparsity of the coefficients and the sparsity of their successive differences for the linear regression. Due to its wide applications, there are many existing algorithms to solve fused Lasso. However, the computation of this model is time-consuming in high-dimensional data sets. To accelerate the calculation of fused Lasso in high-dimension data sets, we build up the safe feature identification rule by introducing an extra dual variable. With a low computational cost, this rule can eliminate inactive features with zero coefficients and identify adjacent features with same coefficients in the solution. To the best of our knowledge, existing screening rules can not be applied to speed up the computation of fused Lasso and our work is the first one to deal with this problem. To emphasize our rule is a unique result that is capable of identifying adjacent features with same coefficients, we name the result as the safe feature identification rule. Numerical experiments on simulation and real data illustrate the efficiency of the rule, which means this rule can reduce the computational time of fused Lasso. In addition, our rule can be embedded into any efficient algorithm and speed up the computational process of fused Lasso.
\section{Introduction}
To simultaneously characterize the sparsity of coefficients and the sparsity of their successive differences for the linear regression, Tibshirani et al. \cite{tibshirani2005sparsity}  proposed fused Lasso, that is
\begin{equation}\label{eq:1}
\underset{\boldsymbol{\beta}\in \mathbb{R}^{p}}\min~ {\frac{1}{2}\|\textbf{\textit{y}}-X\boldsymbol{\beta}\|^{2}_{2}+\lambda_{1}\|\boldsymbol{\beta}\|_{1}+\lambda_{2}\|D\boldsymbol{\beta}\|_{1}},
\end{equation}
where $\|\boldsymbol{\beta}\|_{1}$ is the Lasso term and $\|D\boldsymbol{\beta}\|_{1}$ is the fused term. Here,
\begin{equation}\label{eq:2}
D=\begin{bmatrix}
1&-1&0&0&\cdots&0&0\\
0&1&-1&0&\cdots&0&0\\
\vdots&\vdots&\vdots&\vdots&\vdots&\vdots&\vdots\\
0&0&0&0&\cdots&1&-1
\end{bmatrix}\in\mathbb{R}^{(p-1)\times p}.
\end{equation}
This model is proved to be widely used in many areas, such as prostate cancer diagnosis \cite{tibshirani2005sparsity}, analysis of noisy and imcomplete images \cite{2012Nonparametric}, geonomic analysis \cite{li2008network}, time-varying networks recovery \cite{ahmed2009recovering} and so on. Due to its wide applications, there are many algorithms  developed to solve fused Lasso. For instance,  SLEP \cite{liu2009slep}, EFLA \cite{liu2010efficient}, ADM \cite{he2012alternating}, stochastic ADMM \cite{ouyang2013stochastic} and semismooth Newton based augmented Lagrangian method \cite{li2018efficiently}. However, with the feature size $p$ increasing, these algorithms may become time-consuming.

In order to overcome the computational difficulty of the sparse supervised models  in high-dimensional data sets, some screening rules are proposed to eliminate inactive features or samples. See, e.g., \cite{fan2008sure,el2010safe,tibshirani2012strong,wang2014scaling,wang2015lasso,wang2015fused,shibagaki2016simultaneous,ndiaye2017gap,kuang2017screening, xiang2016screening,lee2017ensembles,ren2017safe,hong2019scaling,pan2018novel,wang2019simultaneous,bao2020fast,chen2020safe,dantas2021expanding,rakotomamonjy2019screening,shang2021ell,pan2021safe}. These screening rules are used to eliminate inactive features for the sparse linear regression or non support vectors (inactive samples) for the support vector machine. Because fused Lasso is the aim, we just review some results for the sparse linear regression in this paper. For example, Ghaoui et al. \cite{el2010safe} constructed SAFE rules for sparse supervised models, which includes the Lasso, sparse  logistic regression and $\ell_{1}$-norm regularized support vector machine. Tibshirani et al. \cite{tibshirani2012strong} proposed strong rules for the Lasso under the assumption of the unit slope bound. Due to that strong rules are heuristic,  the Karush-Kuhn-Tucker (KKT) condition is needed to be checked. Wang et al. \cite{wang2015lasso} proposed  the dual polytope projection (DPP) and EDPP  for Lasso. \textcolor{black}{These screening rules are proposed to speed up the computation of Lasso, which can be seemed as a reduced case of fused Lasso when $\lambda_{2}=0$. For fused Lasso with $\lambda_{1}=0$, Wang et al. \cite{wang2015fused} introduced two affine transformations and set up the screening rule via the monotonicity of the subdifferentials. However, fused Lasso can not be transformed into Lasso form obviously or  transformations in \cite{wang2015fused} do no need to be applied as Remark \ref{remark3.2}.} In addition, Ndiaye et al. \cite{ndiaye2017gap}  built up  statics and dynamic  gap safe screening rules for group Lasso, which are based on the gap between feasible points of group Lasso and its dual problem.  Chen et al. \cite{chen2020safe} proposed the safe screening rules for the regularized Huber regression. Shang et al. \cite{shang2021ell,shang2022safe} introduced the dual circumscribed sphere technique and apply this technique to build up safe feature screening rules for  quantile Lasso and rank Lasso.

Up to now, existing screening rules mainly focus on sparse linear regression with one regularizer. \textcolor{black}{Because the linear regression with two regularizers usually consider more structures of solution, it is more difficult to be analyzed in theory and computation. In addition, to build up the screening rule, the common structure is based on the dual form and estimation of dual solution. The dual form of regularized models with two regularizers usually is more complex in formulation, which makes the estimation of dual solution more complex and the design of screening rule more challenging. To the best of our knowledge, there are a few works that focus on the linear regression with two regularizers. Tibshirani et al. \cite{tibshirani2012strong} and Xu et al. \cite{xu2019endpp} applied the strong rules and DPP to Elastic Net, respectively. Although Elastic Net has two regularizers, it is easy to be transformed as Lasso form, which potentially deal with one regularizer. In this paper, we tend to build up the safe feature screening rule for fused Lasso, so the strong rules in Tibshirani et al. \cite{tibshirani2012strong} does not work. In addition, fused Lasso can not be transformed to a simple form as Lasso. Ren et al. \cite{ren2017safe} proposed a novel bound propagation algorithm for general Lasso problems. In experiments, they transformed (\ref{eq:1}) to the form with just one regularizer $\lambda_{2}\|\tilde{D}\boldsymbol{\beta}\|$, where $\tilde{D}=(\lambda_{1}I_{p}/\lambda_{2};D)$. Some experiment results of their screening rule on model (1) were reported under $\lambda_{1}/\lambda_{2}$ changing from $0.1$ to $10$, which means they potentially deal with only one regularizer. Wang et al. \cite{wang2019two} constructed the screening rule for the sparse-group Lasso based on an important technique that are the union conjugate function of the Lasso term and the group term. With the help of this technique, they presented the dual problem of sparse-group Lasso with a unique variable. As Remark \ref{remark3.1}, the application of this technique in \cite{wang2019two} can not deal with the screening rule for fused Lasso. }

In this paper, we propose the safe feature identification rule for fused Lasso, by introducing an extra dual variable. \textcolor{black}{To achieve the actual two regularizers purpose and highlight our contribution, values of $\lambda_{1}$ and $\lambda_{2}$ in (\ref{eq:1}) are usually larger than zero in theory and computation.} Firstly, we present the dual problem and the KKT system of fused Lasso, with the help of the conjugate function of the Lasso term. Because the conjugate function of the fused term does not have a simple form, we introduce an extra variable for fused Lasso, which makes two variables in its dual problem. According to the KKT system of fused Lasso, we obtain a rule to screen inactive features and identify adjacent features with same coefficients, under the condition of the known dual solution. Secondly, to be implementable, we estimate the dual solution of fused Lasso. From the equation (3) in Section 3.1, we know that the dual problem contains two variables, while its objective function is only related to one variable. This is a difficulty to set up the identification rule for fused Lasso. To tackle this challenge, we estimate the dual solution in two steps. Based on the variational inequality, the variable in the objective function can be estimated. According to this estimation, we can bound the other variable in the dual solution. Thirdly, under the estimations of the dual solution and KKT system, we get the safe identification screening rule, which is composed of two theorems. One theorem is used to eliminate the inactive features with zero coefficient in the solution, and the other is used to identify adjacent features with same coefficients. Finally, we evaluate the screening rule on some simulation data and real data, which shows that our rule is efficient on eliminating inactive features and speeding up the computation of fused Lasso. In the numerical experiments, we use the SLEP to solve fused Lasso. Actually,  our screening rule can be embedded into any efficient algorithm and reduce the computational time.

 The rest of this paper is organized as follows. We review basic concepts and results in Section \ref{sec:pre}.  In Section \ref{sec:method}, we build up the safe feature identification rule for fused Lasso.  In Section \ref{sec:exp}, we illustrate the efficiency of our rule on some simulation data and real data.  Some conclusions are given in Section \ref{sec:con}.
\section{Preliminaries}\label{sec:pre}
In this section, we review some basic definitions and results in \cite{rockafellar1970convex,hiriart2013convex} for this paper.
\begin{Definition}
Let $f:\mathbb{R}^{n}\rightarrow (-\infty,+\infty]$ be a proper closed convex function and let $\textbf{x}\in dom(f)$. A vector $\textbf{g}\in \mathbb{R}^{n}$ is called a subgradient of $f$ at $\textbf{x}$ if
\begin{center}
$f(\textbf{\textit{y}})\geq f(\textbf{x})+\langle \textbf{g},\textbf{\textit{y}}-\textbf{x}\rangle$,  $\forall \textbf{\textit{y}}\in \mathbb{R}^{n}$.
\end{center}
The set of all subgradients of $f$ at $\textbf{x}$ is called the subdifferential of $f$ at $\textbf{x}$ and is denoted by $\partial f(\textbf{x})$, that is
\begin{center}
$\partial f(\textbf{x})=\left\{\textbf{g}\in \mathbb{R}^{n}: f(\textbf{\textit{y}})\geq f(\textbf{x})+\langle \textbf{g},\textbf{\textit{y}}-\textbf{x}\rangle, \forall \textbf{\textit{y}}\in \mathbb{R}^{n} \right\}$.
\end{center}
\end{Definition}
Now, we review the definition of the conjugate function.
\begin{Definition}
Let $f:\mathbb{R}^{n}\rightarrow (-\infty,+\infty]$ be a proper closed convex function. The conjugate function of $f$ is denoted as $f^{*}$ and $f^{*}:\mathbb{R}^{n}\rightarrow (-\infty,+\infty]$  is defined as
\begin{center}
$f^{*}(\textbf{\textit{y}})=\underset{\textbf{x}\in \mathbb{R}^{n}}\max\left\{\langle \textbf{\textit{y}},\textbf{x}\rangle-f(\textbf{x})\right\}$, $\forall \textbf{\textit{y}}\in \mathbb{R}^{n}.$
\end{center}
\end{Definition}

For any $\textbf{\textit{x}}=(x_{1},x_{2},\cdots,x_{n})^{\top}\in \mathbb{R}^{n}$, the $\ell_{1}$ norm of $\textbf{\textit{x}}$ is defined as $\|\textbf{\textit{x}}\|_{1}=\sum\limits_{i=1}^{n}|x_{i}|$.  The $\ell_{\infty}$ norm of $\textbf{\textit{x}}$ is defined as $\|\textbf{\textit{x}}\|_{\infty}=\max\{|x_{1}|,\cdots,|x_{n}|\}$. The $\ell_{2}$ norm of $\textbf{\textit{x}}$ is defined as $\|\textbf{\textit{x}}\|_{2}=\sqrt{x^{2}_{1}+\cdots +x^{2}_{n}}$. According to these definitions, we review some results in the next example.
\begin{Example}
Let $\textbf{x}=(x_{1},x_{2},\cdots,x_{p})^{\top}\in \mathbb{R}^{p}$ and $\lambda>0$.\\
(a) The subdifferential of $\|\textbf{x}\|_{1}$ is

\textcolor{black}{\begin{equation*}
\partial\|\textbf{x}\|_{1}=\left\{\boldsymbol{\omega}\in\mathbb{R}^{p}:
\omega_{j}\in\rm{sign}(\textit{x}_{\textit{j}})=
\begin{cases}
\{1\},&x_{j}>0\\
\left[-1,1\right],&x_{j}=0\\
\{-1\},&x_{j}<0
\end{cases}, \textit{j}=1,2,\cdots,\textit{p}\right\}.
\end{equation*}}
(b) The conjugate function of $\|\textbf{x}\|_{1}$ is the indictor function of $\ell_{\infty}$ norm unit ball, i.e.,
\begin{equation*}
\left(\|\textbf{x}\|_{1}\right)^{*}=\underset{\textbf{z}\in \mathbb{R}^{p}}\max\left\{\langle \textbf{x},\textbf{z}\rangle-\|\textbf{z}\|_{1}\right\}=\delta_{\|\cdot\|_{\infty}\leq1}(\textbf{x})=
\begin{cases}
0,&\|\textbf{x}\|_{\infty}\leq1;\\
+\infty, &\rm{otherwise}.
\end{cases}
\end{equation*}
\end{Example}

It is worth to mention that the notation 0 in this paper may be a scalar, vector and matrix. One can infer its detailed meaning based on the context.

\section{Safe feature identification rule for fused Lasso}\label{sec:method}
In this section, we show the duality theory of fused Lasso and  propose the safe feature identification rule, which identify inactive features and  adjacent features with same coefficients.
\subsection{Duality theory of fused Lasso}
Fused Lasso is reviewed as
\begin{eqnarray*}
\underset{\boldsymbol{\beta}\in \mathbb{R}^{p}}\min~ {\frac{1}{2}\|\textbf{\textit{y}}-X\boldsymbol{\beta}\|^{2}_{2}+\lambda_{1}\|\boldsymbol{\beta}\|_{1}+\lambda_{2}\|D\boldsymbol{\beta}\|_{1}},
\end{eqnarray*}
where $\boldsymbol{\beta}=(\beta_{1},\cdots, \beta_{p})^{\top}\in \mathbb{R}^{p}$ is the unknown coefficient vector, $\textbf{\textit{y}}=(y_{1},\cdots,y_{n})^{\top}\in \mathbb{R}^{n}$ is the response variable and  $X=(X_{\cdot1},\cdots,X_{\cdot p})\in \mathbb{R}^{n\times p}$ is the prediction matrix with $X_{\cdot j}\in \mathbb{R}^{n}$ denoting the $j_{th}$ feature.  Here, $\lambda_{1}>0$ and $\lambda_{2}>0$ are tuning parameters. The Lasso term $\lambda_{1}\|\boldsymbol{\beta}\|_{1}$ induces the sparsity of coefficients and the fusion term $\lambda_{2}\|D\boldsymbol{\beta}\|_{1}$ encourages the sparsity of their differences. It is clear that the matrix $D$ has full row rank, which infers that the inverse matrix of $DD^{\top}$ exists. If $\lambda_{2}$ is fixed, a larger $\lambda_{1}$ leads to a more sparse solution. If $\lambda_{1}$ is fixed, a larger $\lambda_{2}$ leads to a solution with more successively same elements. Because the solution of (\ref{eq:1}) relies on the choices of $\lambda_{1}$ and $\lambda_{2}$, we denote it as $\boldsymbol{\beta}^{*}(\lambda_{1},\lambda_{2})$.

With the help of Example 2.1 and the Lagrangian duality theory, we prove the next result.
\begin{Lemma}
Let $\textbf{x}=(x_{1},x_{2},\cdots,x_{p})^{\top}\in \mathbb{R}^{p}$ and $D\in\mathbb{R}^{(p-1)\times p }$  defined in (\ref{eq:2}).\\
(a) The subdifferential of $\|D\textbf{x}\|_{1}$ is
$\partial\|D\textbf{x}\|_{1}=D^{\top}\partial_{D\textbf{x}}\|D\textbf{x}\|_{1}$, i.e.,
\textcolor{black}{\begin{equation}\label{eq: subdifferential_fused_lasso}
\begin{aligned}
&\partial\|D\textbf{x}\|_{1}\\
&=\left\{\boldsymbol{\omega}\in\mathbb{R}^{p}:\omega_{j}\in
\begin{cases}
\rm{sign}(\textit{x}_{\textit{j}}-\textit{x}_{\textit{j}+1}), &j=1\\
-\rm{sign}(\textit{x}_{\textit{j}-1}-\textit{x}_{\textit{j}}), &j=p\\
\rm{sign}(\textit{x}_{\textit{j}}-\textit{x}_{\textit{j}+1})-\rm{sign}(\textit{x}_{\textit{j}-1}-\textit{x}_{\textit{j}}),
&\rm{otherwise}
\end{cases}\right\}.
\end{aligned}
\end{equation}}
(b) The conjugate function of $\|D\textbf{x}\|_{1}$ is showed as follows.
\begin{equation}\label{eq: conjugate_fused_lasso}
\left(\|D\textbf{x}\|_{1}\right)^{*}=
\begin{cases}
0,& \left\|(DD^{\top})^{-1}D\textbf{\textit{x}}\right\|_{\infty}\leq 1;\\
+\infty, &\rm{otherwise}.
\end{cases}
\end{equation}
\end{Lemma}
\begin{proof}
The result (a) can be easily obtained with the chain rule. Here, we present the proof of the result (b). According to Definition 2.2, we know that
\begin{equation}\tag{*}
\begin{split}
\left(\|D\textbf{\textit{x}}\|_{1}\right)^{*}&=
\underset{\textbf{\textit{z}}\in\mathbb{R}^{p}}\max\left\{\langle\textbf{\textit{z}},\textbf{\textit{x}}\rangle-\|D\textbf{\textit{z}}\|_{1}\right\}\\
&=\underset{\textbf{\textit{z}}\in\mathbb{R}^{p},\textbf{\textit{w}}\in\mathbb{R}^{p-1}}
\max\left\{\langle\textbf{\textit{z}},\textbf{\textit{x}}\rangle
-\|\textbf{\textit{w}}\|_{1}: D\textbf{\textit{z}}-\textbf{\textit{w}}=0\right\}.
\end{split}
\end{equation}
By introducing the Lagrangian multiplier $\boldsymbol{\theta}\in\mathbb{R}^{p-1}$, the Lagrangian function of (*) is
\begin{center}
$\textit{L}\left(\textbf{\textit{z}},\textbf{\textit{w}}, \boldsymbol{\theta}\right)
=\langle\textbf{\textit{z}},\textbf{\textit{x}}\rangle-\|\textbf{\textit{w}}\|_{1}
+\langle D\textbf{\textit{z}}-\textbf{\textit{w}},\boldsymbol{\theta}\rangle$
\end{center}
and
\begin{align*}
&\underset{\textbf{\textit{z}}\in\mathbb{R}^{p},\textbf{\textit{w}}\in\mathbb{R}^{p-1}}\max\textit{L}\left(\textbf{\textit{z}},\textbf{\textit{w}}, \boldsymbol{\theta}\right)=\underset{\textbf{\textit{z}}\in\mathbb{R}^{p}}\max \left\{\langle \textbf{\textit{x}}+D^{\top}\boldsymbol{\theta},\textbf{\textit{z}}\rangle\right\}
+\underset{\textbf{\textit{w}}\in\mathbb{R}^{p-1}}\max \left\{\langle\textbf{\textit{w}},-\boldsymbol{\theta}\rangle-\|\textbf{\textit{w}}\|_{1}\right\}\\
&\quad\quad =\begin{cases}
0, &\textbf{\textit{x}}+D^{\top}\boldsymbol{\theta}=0, \|\boldsymbol{\theta}\|_{\infty}\leq 1,\\
+\infty, &\rm{otherwise}.
\end{cases}
\end{align*}
According to this result, the dual problem of (*) is
\begin{align*}
&\underset{\boldsymbol{\theta}}\min~\underset{\textbf{\textit{z}},\textbf{\textit{w}}}\max~\textit{L}\left(\textbf{\textit{z}},\textbf{\textit{w}}, \boldsymbol{\theta}\right)=\begin{cases}
0, &\left\|(DD^{\top})^{-1}D\boldsymbol{x}\right\|_{\infty}\leq1,\\
+\infty, &\rm{otherwise}.
\end{cases}
\end{align*}
Because there is a feasible point $(\textbf{\textit{z}},\textbf{\textit{w}})=(0,0)$ for the problem (*), the Slater's constraint qualification holds and
\begin{center}
$\left(\|D\textbf{\textit{x}}\|_{1}\right)^{*}=
\underset{\boldsymbol{\theta}}\min~
\underset{\textbf{\textit{z}},\textbf{\textit{w}}}\max~\textit{L}\left(\textbf{\textit{z}},\textbf{\textit{w}}, \boldsymbol{\theta}\right)$.
\end{center}
The desired result follows.
\end{proof}

In order to get the dual problem of (\ref{eq:1}), we introduce two variables $\boldsymbol{\alpha}\in \mathbb{R}^{n}$, $\boldsymbol{\gamma}\in \mathbb{R}^{p-1}$ and transform (\ref{eq:1}) to a constraint problem as below.
\begin{equation*}
\begin{split}
\underset{\boldsymbol{\beta}\in\mathbb{R}^{p},\boldsymbol{\alpha}\in\mathbb{R}^{n},\boldsymbol{\gamma}\in\mathbb{R}^{p-1}}\min&\frac{1}{2}\|\boldsymbol{\alpha}\|^{2}_{2}+\lambda_{1}\|\boldsymbol{\beta}\|_{1}+\lambda_{2}\|\boldsymbol{\gamma}\|_{1}\\
s.t.~~~~~~~~~~~&\textbf{\textit{y}}-X\boldsymbol{\beta}-\boldsymbol{\alpha} = 0, \\
&  D\boldsymbol{\beta}-\boldsymbol{\gamma}=0.
\end{split}
\end{equation*}
By introducing Lagrangian multipliers $\textbf{\textit{u}}=(\textit{u}_{1},\textit{u}_{2},\cdots,\textit{u}_{n})^{\top}\in \mathbb{R}^{n}$ and $\textbf{\textit{v}}=(\textit{v}_{1},\textit{v}_{2},\cdots,\textit{v}_{p-1})^{\top}\in \mathbb{R}^{p-1}$,
 we obtain the Lagrangian function of this model, which is
\begin{equation*}
\begin{aligned}
&\textit{L}\left(\boldsymbol{\beta},\boldsymbol{\alpha},\boldsymbol{\gamma}, \textbf{\textit{u}},\textbf{\textit{v}}\right)\\
&=\frac{1}{2}\|\boldsymbol{\alpha}\|^{2}_{2}+\lambda_{1}\|\boldsymbol{\beta}\|_{1}+\lambda_{2}\|\boldsymbol{\gamma}\|_{1}+\langle\textbf{\textit{u}},\textbf{\textit{y}}-X\boldsymbol{\beta}-\boldsymbol{\alpha}\rangle
+\langle\textbf{\textit{v}},D\boldsymbol{\beta}-\boldsymbol{\gamma}\rangle.
\end{aligned}
\end{equation*}
According to the results in Example 2.1, we obtain that
\begin{align*}
&\underset{\boldsymbol{\beta}\in\mathbb{R}^{p},\boldsymbol{\alpha}\in\mathbb{R}^{n},\boldsymbol{\gamma}\in\mathbb{R}^{p-1}}\min\textit{L}\left(\boldsymbol{\beta},\boldsymbol{\alpha},\boldsymbol{\gamma}, \textbf{\textit{u}},\textbf{\textit{v}}\right)\\
&=\underset{\boldsymbol{\beta}\in\mathbb{R}^{p}}\min \left\{\lambda_{1}\|\boldsymbol{\beta}\|_{1}-\langle X^{\top}\textbf{\textit{u}}-D^{\top}\textbf{\textit{v}},\boldsymbol{\beta}\rangle\right\}
+\underset{\boldsymbol{\alpha}\in\mathbb{R}^{n}}\min \left\{\frac{1}{2}\|\boldsymbol{\alpha}\|^{2}_{2}-\langle\textbf{\textit{u}},\boldsymbol{\alpha}\rangle\right\}\\
&+
\underset{\boldsymbol{\gamma}\in\mathbb{R}^{p-1}}\min \left\{\lambda_{2}\|\boldsymbol{\gamma}\|_{1}-\langle \textbf{\textit{v}},\boldsymbol{\gamma}\rangle\right\}
+\langle\textbf{\textit{u}},\textbf{\textit{y}}\rangle\\
&=-\delta_{\|\cdot\|_{\infty}\leq\lambda_{1}}\left(X^{\top}\textbf{\textit{u}}-D^{\top}\textbf{\textit{v}}\right)
-\delta_{\|\cdot\|_{\infty}\leq\lambda_{2}}\left(\textbf{\textit{v}}\right)-\frac{1}{2}\|\textbf{\textit{u}}\|^{2}_{2}+\langle \textbf{\textit{u}},\textbf{\textit{y}}\rangle.
\end{align*}
Therefore, the Lagrangian dual form of the model (\ref{eq:1}) is $$\underset{\textbf{\textit{u}},\textbf{\textit{v}}}\max~
\underset{\boldsymbol{\beta},\boldsymbol{\alpha},\boldsymbol{\gamma}}\min~\textit{L}\left(\boldsymbol{\beta},\boldsymbol{\alpha},\boldsymbol{\gamma}, \textbf{\textit{u}},\textbf{\textit{v}}\right).$$
By a simple computation, we get the dual problem of (\ref{eq:1}) as
\begin{equation}\label{eq:3}
\begin{split}
\underset{\textbf{\textit{u}}\in \mathbb{R}^{n},\textbf{\textit{v}}\in \mathbb{R}^{p-1}}\min  &\frac{1}{2}\left\|\textbf{\textit{u}}\right\|^{2}_{2}-\langle\textbf{\textit{u}},\textbf{\textit{y}}\rangle\\
s.t. \quad &\|X^{\top}\textbf{\textit{u}}-D^{\top}\textbf{\textit{v}}\|_{\infty}\leq \lambda_{1}, \\
\quad \quad ~ &\|\textbf{\textit{v}}\|_{\infty}\leq\lambda_{2}.
\end{split}
\end{equation}
Because the solutions of this problem rely on the choices of $\lambda_{1}$ and $\lambda_{2}$, we denote them as $\left(\textbf{\textit{u}}^{*}(\lambda_{1},\lambda_{2}),\textbf{\textit{v}}^{*}(\lambda_{1},\lambda_{2})\right)$.

The KKT system of (\ref{eq:1}) is
\begin{equation}\label{eq:KKT}
\begin{cases}
X^{\top}\textbf{\textit{u}}-D^{\top}\textbf{\textit{v}}\in \lambda_{1}\partial\|\boldsymbol{\beta}\|_{1},\\
\textbf{\textit{v}}\in \lambda_{2}\partial\|\boldsymbol{\gamma}\|_{1},\\
\textbf{\textit{y}}-X\boldsymbol{\beta}-\textbf{\textit{u}}=0, D\boldsymbol{\beta}-\boldsymbol{\gamma}=0.
\end{cases}
\end{equation}
Any pair $(\boldsymbol{\beta},\textbf{\textit{u}},\textbf{\textit{v}})$ satisfies the KKT system, it is a KKT point. According to the Slater's constraint qualification, it is easy to obtain that  solutions of (\ref{eq:1}) and (\ref{eq:3}) compose a KKT point.

\begin{Remark}\label{remark3.1}
Here, we show some analysis of the dual problem (\ref{eq:3}). The dual problem (\ref{eq:3}) is obtained by introducing variables $\boldsymbol{\alpha}$ and $\boldsymbol{\gamma}$ such that $\textbf{\textit{y}}-X\boldsymbol{\beta}-\boldsymbol{\alpha} =0$  and $D\boldsymbol{\beta}-\boldsymbol{\gamma}=0$. It is easy to get the dual form remains same if  $\boldsymbol{\alpha}$ and $\boldsymbol{\gamma}$ satisfy that $\textbf{\textit{y}}-X\boldsymbol{\beta}-\boldsymbol{\alpha}=0$ and $\boldsymbol{\beta}-\boldsymbol{\gamma}=0$.  For this problem, there are two variables $\textbf{\textit{u}}$ and $\textbf{\textit{v}}$, while the objective function only contains one variable $\textbf{\textit{u}}$. The variable $\textbf{\textit{v}}$ is introduced to restrict the feasible set of $\textbf{\textit{u}}$. This is a unique phenomenon for fused Lasso. For the problems with one regularizer, this phenomenon does not exists because there is only one dual variable, which means the existing screening rules for the linear regression with one regularizer can not speed up the computation of fused Lasso.

For the problems with two  regularizer, such as sparse-group Lasso, Elastic Net, $\ell_{1}$ regularized SVM and so on, \cite{wang2019two} proposed a method to give the Fenchel's dual problem that only contains one dual variable. Although the Fenchel's dual problem of fused Lasso can also be showed as that of \cite{wang2019two}, the complex form of the subdifferential of the fusion term in Lemma 3.1 makes the screening rule in \cite{wang2019two} can not be applied to fused Lasso. There is a detailed analysis as follows.

With Example 2.1, Lemma 3.1 and \cite[Theorem 4.17]{beck2017first},
\begin{equation}
\begin{aligned}
\left(\lambda_{1}\|\boldsymbol{\beta}\|+\lambda_{2}\|D\boldsymbol{\beta}\|\right)^{*}
&=\delta_{\|\cdot\|_{\infty}\leq \lambda_{1}}(\boldsymbol{\beta})\square \delta_{\|(DD^{\top})^{-1}D\cdot\|_{\infty}\leq \lambda_{2}}(\boldsymbol{\beta})\\
&=\underset{\boldsymbol{\gamma}}\inf\left\{\delta_{\|\cdot\|_{\infty}\leq \lambda_{1}}(\boldsymbol{\gamma})+\delta_{\left\|(DD^{\top})^{-1}D\cdot\right\|_{\infty}\leq \lambda_{2}}(\boldsymbol{\beta}-\boldsymbol{\gamma})\right\}\\
&=\delta_{\|\cdot\|_{\infty}\leq \lambda_{2}}\left((DD^{\top})^{-1}D(\boldsymbol{\beta}-\Pi_{\|\cdot\|\leq\lambda_{1}}(\boldsymbol{\beta}))\right),
\end{aligned}
\end{equation}
where $\Pi_{\|\cdot\|\leq\lambda_{1}}(\boldsymbol{\beta})$ is the projection operator and $\square$ is the infimal convolution. Under this formulation, the dual problem of fused Lasso can be expressed as
\begin{equation*}
\begin{aligned}
\underset{\boldsymbol{\theta}\in \mathbb{R}^{n}}\min ~ &\frac{1}{2}\left\|\boldsymbol{\theta}\right\|^{2}_{2}-\langle\boldsymbol{\theta},\textbf{\textit{y}}\rangle\\
s.t. ~ &\left\|(DD^{\top})^{-1}D(X^{\top}\boldsymbol{\theta}-\Pi_{\|\cdot\|\leq\lambda_{1}}(X^{\top}\boldsymbol{\theta}))\right\|_{\infty}\leq \lambda_{2}.
\end{aligned}
\end{equation*}
In this case, the KKT system is
\begin{equation*}
\begin{cases}
X^{\top}\boldsymbol{\theta}\in \lambda_{1}\partial\|\boldsymbol{\beta}\|_{1}+\lambda_{2}\partial\|D\boldsymbol{\beta}\|_{1},\\
\textbf{\textit{y}}-X\boldsymbol{\beta}-\boldsymbol{\theta}=0.
\end{cases}
\end{equation*}
The above results are all followed the analysis in  \cite{wang2019two}, however these new KKT system can not provide specific criteria about the zero elements or adjacent equivalent elements of  of $\boldsymbol{\beta}^{*}(\lambda_{1},\lambda_{2})$. This is the core difference between our work and  \cite{wang2019two}. The pattern in \cite{wang2019two} fails for fused Lasso mainly due to the subdifferential result of fused term not as simple as $\ell_{1}$-norm or grouped $\ell_{2}$-norm.
\end{Remark}

Based on KKT system, subdifferentials of the $\ell_{1}$ norm in Example 2.1 and fused term in Lemma 3.1, we can get the following lemma. Unlike other models, there is a unique phenomenon of fused Lasso as in (ii) of this lemma.
\begin{Lemma}
Let $\lambda_{1}\geq0$ and $\lambda_{2}\geq0$.\\
(i) For any $j\in\{1,2,\cdots,p\}$, $\beta^{*}_{j}(\lambda_{1},\lambda_{2})=0$ if
\begin{center}
$\Big|X_{.j}^{\top}\textbf{\textit{u}}^{*}(\lambda_{1},\lambda_{2})-D_{.j}^{\top}\textbf{\textit{v}}^{*}(\lambda_{1},\lambda_{2})\Big|<\lambda_{1}$.
\end{center}
(ii) For any $j\in\{1,2,\cdots,p-1\}$, $\beta^{*}_{j}(\lambda_{1},\lambda_{2})=\beta^{*}_{j+1}(\lambda_{1},\lambda_{2})$ if
\begin{center}
$\big|v_{j}^{*}(\lambda_{1},\lambda_{2})\big|<\lambda_{2}$.
\end{center}
\end{Lemma}
\begin{proof}
(i) Because that $\textbf{\textit{u}}^{*}(\lambda_{1},\lambda_{2})$ and $\textbf{\textit{v}}^{*}(\lambda_{1},\lambda_{2})$ satisfy the KKT system, it holds
\begin{center}
$X^{\top}\textbf{\textit{u}}^{*}(\lambda_{1},\lambda_{2})-D^{\top}\textbf{\textit{v}}^{*}(\lambda_{1},\lambda_{2})\in \lambda_{1}\partial\|\boldsymbol{\beta}^{*}(\lambda_{1},\lambda_{2})\|_{1}$,
\end{center}
which leads to
\begin{center}
$X_{.j}^{\top}\textbf{\textit{u}}^{*}(\lambda_{1},\lambda_{2})-D_{.j}^{\top}\textbf{\textit{v}}^{*}(\lambda_{1},\lambda_{2})
\in\lambda_{1}$sign$\left(\beta^{*}_{j}(\lambda_{1},\lambda_{2})\right)$.
\end{center}
Then, the desired result is obtained by the subdifferential of the $\ell_{1}$ norm in Example 2.1.

(ii) From the KKT system again, we know that
\begin{center}
$\textbf{\textit{v}}^{*}(\lambda_{1},\lambda_{2})\in \lambda_{2}\partial \left\|\boldsymbol{\gamma}^{*}(\lambda_{1},\lambda_{2})\right\|_{1}$,
$D\boldsymbol{\beta}^{*}(\lambda_{1},\lambda_{2})=\boldsymbol{\gamma}^{*}(\lambda_{1},\lambda_{2})$,
\end{center}
which leads to $\textbf{\textit{v}}^{*}(\lambda_{1},\lambda_{2})\in \lambda_{2}\partial_{D\boldsymbol{\beta}^{*}(\lambda_{1},\lambda_{2})} \left\|D\boldsymbol{\beta}^{*}(\lambda_{1},\lambda_{2})\right\|_{1}$. That is
\begin{center}
$\textit{v}_{j}^{*}(\lambda_{1},\lambda_{2})\in\lambda_{2}$sign$\left(\beta^{*}_{j}(\lambda_{1},\lambda_{2})-\beta^{*}_{j+1}(\lambda_{1},\lambda_{2})\right)$.
\end{center}
Combining Example 2.1, the desired result is obtained.
\end{proof}
\begin{Remark}\label{remark3.2}
In \cite{wang2015fused}, they introduced an affine transformation and transform the dual problem of their model as a projection problem. The transformation is
\begin{equation*}
T=\begin{bmatrix}
1 & 1 & 1 &\cdots &1\\
0 & 1 & 1 &\cdots &1\\
0 & 0 & 1 &\cdots &1\\
\vdots & \vdots &\vdots &\ddots &\vdots\\
0 & 0 & 0 &\cdots &1
\end{bmatrix}.
\end{equation*}
Following the computation in \cite{wang2015fused}, with the matrix $T$, it holds
\begin{equation}\label{eq: wrong}
\begin{aligned}
T\textbf{\textit{v}}^{*}(\lambda_{1},\lambda_{2})
&=\left(\sum\limits_{j=1}^{p-1}\textit{v}_{j}^{*}(\lambda_{1},\lambda_{2}), \sum\limits_{j=2}^{p-1}\textit{v}_{j}^{*}(\lambda_{1},\lambda_{2}),\cdots, \textit{v}_{p-1}^{*}(\lambda_{1},\lambda_{2})\right)^{\top}\\
&=\lambda_{2}
\begin{pmatrix}
\sum\limits_{j=1}^{p-1}\rm{sign}\left(\beta_{j}^{*}(\lambda_{1},\lambda_{2})-\beta_{j+1}^{*}(\lambda_{1},\lambda_{2})\right)\\
\sum\limits_{j=2}^{p-1}\rm{sign}\left(\beta_{j}^{*}(\lambda_{1},\lambda_{2})-\beta_{j+1}^{*}(\lambda_{1},\lambda_{2})\right)\\
\vdots\\
\rm{sign}\left(\beta_{\textit{p}-1}^{*}(\lambda_{1},\lambda_{2})-\beta_{\textit{p}}^{*}(\lambda_{1},\lambda_{2})\right)
\end{pmatrix}.
\end{aligned}
\end{equation}
Based on the equation (\ref{eq: wrong}), the following results hold.

(i) $\beta_{\textit{p}-1}^{*}(\lambda_{1},\lambda_{2})=\beta_{\textit{p}}^{*}(\lambda_{1},\lambda_{2})$ holds, if $\textit{v}_{p-1}^{*}(\lambda_{1},\lambda_{2})<\lambda_{2}$.

(ii) $\beta_{\textit{p}-2}^{*}(\lambda_{1},\lambda_{2})=\beta_{\textit{p}-1}^{*}(\lambda_{1},\lambda_{2})$ holds, if $\big|\textit{v}_{p-1}^{*}(\lambda_{1},\lambda_{2})\big|=\lambda_{2}$ and $0<\big|\textit{v}_{p-2}^{*}(\lambda_{1},\lambda_{2})\big|<2\lambda_{2}$.

(iii) For any $j\in\{p-3,p-4,\cdots,2,1\}$, the condition of $\beta_{\textit{j}}^{*}(\lambda_{1},\lambda_{2})=\beta_{\textit{j}+1}^{*}(\lambda_{1},\lambda_{2})$ will be more complex.

So, there is no need to apply the technique in \cite{wang2015fused} to accelerate the computation of fused Lasso.
\end{Remark}

The result $\beta^{*}_{j}(\lambda_{1},\lambda_{2})=0$ means the feature $X_{.j}$ is uncorrelated with $\textbf{\textit{y}}$. $\beta^{*}_{j}(\lambda_{1},\lambda_{2})=\beta^{*}_{j+1}(\lambda_{1},\lambda_{2})$ leads to a conclusion that we only need to calculate one of these two values. So $X_{.j}$ can be eliminated if $\beta^{*}_{j}(\lambda_{1},\lambda_{2})=0$, when solving (\ref{eq:1}) under $\lambda_{1}$ and $\lambda_{2}$. Features $X_{.j}$ and $X_{.j+1}$ can be added, the $j$-th row of $D$ can be eliminated, the column number of $X$ and dimension of $\boldsymbol{\beta}^{*}(\lambda_{1},\lambda_{2})$ can be reduced by 1, if the result $\beta^{*}_{j}(\lambda_{1},\lambda_{2})=\beta^{*}_{j+1}(\lambda_{1},\lambda_{2})$ holds. This lemma shows the basic idea that can be used to identify inactive features and adjacent features with same coefficients. However, this result is based on the known solutions of (\ref{eq:3}), which usually needs to be calculated with a proper algorithm. Therefore, we \textcolor{black}{need} to estimate the solutions of (\ref{eq:3}) with a low computational cost.
\subsection{Estimation of the dual solution}
When $\lambda_{1}$ is sufficiently large,  the solution $\boldsymbol{\beta}^{*}(\lambda_{1},\lambda_{2})$ has to be 0 to attain the minimal value of the objective function of (\ref{eq:1}). Similarly, $D\boldsymbol{\beta}^{*}(\lambda_{1},\lambda_{2})$ has to be 0 if $\lambda_{2}$ is large enough.  In the next lemma, we present the lower bounds of $\lambda_{1}$ and $\lambda_{2}$ to achieve these.
\begin{Lemma}
(i) Let $\lambda_{2}>0$ be any fixed constant. If $\boldsymbol{\beta}^{*}(\lambda_{1},\lambda_{2})=0$, we know that $\lambda_{1}\geq\lambda^{max}_{1}(\lambda_{2})$,  where
\textcolor{black}{
\begin{equation}
\lambda^{max}_{1}(\lambda_{2})
=\max\left\{2\lambda_{2}+\kappa_{1},\lambda_{2}+\kappa_{2}\right\},
\end{equation}
$\kappa_{1}=\max\{|X^{\top}_{.j}\textbf{\textit{y}}|:j\in\{2,\cdots,p-1\}\}$ and $\kappa_{2}=\max\{|X^{\top}_{.j}\textbf{\textit{y}}|:j\in\{1,p\}\}$.} On the other hand, 0 is a solution of (\ref{eq:1}) if $\lambda_{1}\geq\lambda^{max}_{1}(\lambda_{2})$.

(ii) Let $\lambda_{1}\geq0$ be any fixed constant. If $\boldsymbol{\beta}^{*}(\lambda_{1},\lambda_{2})=\xi\textbf{1}_{p}$ with $\xi$ being a nonzero constant, i.e., $D\boldsymbol{\beta}^{*}(\lambda_{1},\lambda_{2})=0$,  we know that $\lambda_{2}\geq\lambda^{max}_{2}(\lambda_{1})$, where
\begin{equation}
\begin{aligned}
&\lambda^{max}_{2}(\lambda_{1})\\
&~=
\max\begin{pmatrix}
\underset{j\in\{1,2,\cdots,p-1\}}\max\Big|\sum\limits_{k=1}^{j}\{X^{\top}_{.k}\textbf{\textit{y}}-(X^{\top}X)_{k.}\textbf{1}_{p}\}-j\lambda_{1}\rm{sgn}(\xi)\Big|,\\
|X^{\top}_{.p}\textbf{\textit{y}}-\xi(X^{\top}X)_{p.}\textbf{1}_{p}-\lambda_{1}\rm{sgn}(\xi)|
\end{pmatrix}.
\end{aligned}
\end{equation}
On the other hand, there is a solution of (\ref{eq:1}) such that all of its elements are same if $\lambda_{2}\geq\lambda^{max}_{2}(\lambda_{1})$. Here, $\rm{sgn}(\xi)$ is 1 if $\xi>0$. Otherwise, $\rm{sgn}(\xi)$ is -1.
\end{Lemma}
\begin{proof}
(\textit{i}) If $\boldsymbol{\beta}^{*}(\lambda_{1},\lambda_{2})=0$, according to the KKT system (\ref{eq:KKT}), we have the following results.
$D\boldsymbol{\beta}^{*}(\lambda_{1},\lambda_{2})=0, \textbf{\textit{y}}=\textbf{\textit{u}}^{*}(\lambda_{1},\lambda_{2})$,
 $X^{\top}\textbf{\textit{y}}-D^{\top}\textbf{\textit{v}}^{*}(\lambda_{1},\lambda_{2})\in\lambda_{1}\partial\|0\|_{1}$ and $\textbf{\textit{v}}^{*}(\lambda_{1},\lambda_{2})\in\lambda_{2}\partial\|0\|_{1}$. So, $\lambda_{1}$ and $\lambda_{2}$ need to satisfy that
 \begin{center}
 $\lambda_{1}\geq\underset{j}\max\big|X^{\top}_{.j}\textbf{\textit{y}}-D^{\top}_{.j}\textbf{\textit{v}}^{*}(\lambda_{1},\lambda_{2})\big|$ and $\lambda_{2}\geq\|\textbf{\textit{v}}^{*}(\lambda_{1},\lambda_{2})\|_{\infty}$.
 \end{center}
 Next, we compute
 \textcolor{black}{
 \begin{align*}
 &\underset{\textbf{\textit{v}}\in\mathbb{R}^{p}}
 \max\left\{|X^{\top}_{.j}\textbf{\textit{y}}-D^{\top}_{.j}\textbf{\textit{v}}|
 :\|\textbf{\textit{v}}\|_{\infty}\leq\lambda_{2}\right\}\\
 &=\max\left\{\underset{\textbf{\textit{v}}\in\mathbb{R}^{p}}\max \left\{ X^{\top}_{.j}\textbf{\textit{y}}-D^{\top}_{.j}\textbf{\textit{v}}
 :\|\textbf{\textit{v}}\|_{\infty}\leq\lambda_{2}\right\},
 \underset{\textbf{\textit{v}}\in\mathbb{R}^{p}}\max \left\{ -X^{\top}_{.j}\textbf{\textit{y}}+D^{\top}_{.j}\textbf{\textit{v}}:
 \|\textbf{\textit{v}}\|_{\infty}\leq\lambda_{2}\right\}\right\}.
 \end{align*}}
According to the elements of $D$ in (\ref{eq:2}), we know that
  \textcolor{black}{\begin{align*}
 &\underset{\textbf{\textit{v}}\in\mathbb{R}^{p}}\max \left\{ X^{\top}_{.j}\textbf{\textit{y}}-D^{\top}_{.j}\textbf{\textit{v}}
 :\|\textbf{\textit{v}}\|_{\infty}\leq\lambda_{2}\right\}
 =X^{\top}_{.j}\textbf{\textit{y}}
 +\underset{\textbf{\textit{v}}\in\mathbb{R}^{p}}\max \left\{-D^{\top}_{.j}\textbf{\textit{v}}:\|\textbf{\textit{v}}\|_{\infty}\leq\lambda_{2}\right\}\\
 &=\begin{cases}
 X^{\top}_{.j}\textbf{\textit{y}}+
 \max\left\{-v_{1}:|v_{1}|\leq\lambda_{2}\right\}, &j=1\\
 X^{\top}_{.j}\textbf{\textit{y}}
 +\max\left\{-v_{j}+v_{j-1}:|v_{j-1}|\leq\lambda_{2},|v_{j}|\leq\lambda_{2}\right\}, &\rm{otherwise},\\
 X^{\top}_{.p}\textbf{\textit{y}}+\max \left\{v_{p}:|v_{p}|\leq\lambda_{2}\right\}, &j=p
 \end{cases}\\
&=\begin{cases}
X^{\top}_{.j}\textbf{\textit{y}}+2\lambda_{2}, &j\in\{2,\cdots,p-1\}  \\
X^{\top}_{.j}\textbf{\textit{y}}+\lambda_{2}, &j\in\{1,p\} \end{cases}
 .
\end{align*}}
Similarly,
 \begin{align*}
&\textcolor{black}{\underset{\textbf{\textit{v}}\in\mathbb{R}^{p}}\max \left\{-X^{\top}_{.j}\textbf{\textit{y}}+D^{\top}_{.j}\textbf{\textit{v}}
:\|\textbf{\textit{v}}\|_{\infty}\leq\lambda_{2}\right\}}
=\begin{cases}
-X^{\top}_{.j}\textbf{\textit{y}}+2\lambda_{2}, &j\in\{2,\cdots,p-1\}  \\
-X^{\top}_{.j}\textbf{\textit{y}}+\lambda_{2}, &j\in\{1,p\} \end{cases}.
\end{align*}
So, we get that
\begin{align*}
 &\textcolor{black}{\underset{\textbf{\textit{v}}\in\mathbb{R}^{p}}
 \max\left\{|X^{\top}_{.j}\textbf{\textit{y}}-D^{\top}_{.j}\textbf{\textit{v}}|:
 \|\textbf{\textit{v}}\|_{\infty}\leq\lambda_{2}\right\}}
 =\begin{cases}
 |X_{.j}^{\top}\textbf{\textit{y}}|+2\lambda_{2},& j\in\{2,\cdots,p-1\}\\
|X_{.j}^{\top}\textbf{\textit{y}}|+\lambda_{2}, & j\in\{1,p\}
 \end{cases}.
\end{align*}
and
\begin{align*}
&\textcolor{black}{\underset{j}
\max\left\{|X^{\top}_{.j}\textbf{\textit{y}}-D^{\top}_{.j}\textbf{\textit{v}}^{*}(\lambda_{1},\lambda_{2})|\right\}
=\underset{j}\max\left\{|X^{\top}_{.j}\textbf{\textit{y}}-D^{\top}_{.j}\textbf{\textit{v}}:\|\textbf{\textit{v}}\|_{\infty}\leq\lambda_{2}\right\}}\\
&=\max\left\{|X^{\top}_{.j}\textbf{\textit{y}}|+2\lambda_{2}:j\in\{2,\cdots,p-1\},
|X^{\top}_{.j}\textbf{\textit{y}}|+\lambda_{2}:j\in\{1,p\}\right\}\\
&=\max\left\{2\lambda_{2}+\kappa_{1},\lambda_{2}+\kappa_{2}\right\}\\
&=\lambda^{max}_{1}(\lambda_{2}).
\end{align*}
Therefore, $\lambda_{1}\geq\lambda^{max}_{1}(\lambda_{2})$.
 %\textcolor{black}{Here, notation $|X^{\top}_{.j}\textbf{\textit{y}}|+2\lambda_{2}\Big|_{j\in\{2,\cdots,p-1\}}$ means values of $|X^{\top}_{.2}\textbf{\textit{y}}|+2\lambda_{2}, |X^{\top}_{.3}\textbf{\textit{y}}|+2\lambda_{2},\cdots $ and $|X^{\top}_{.p-1}\textbf{\textit{y}}|+2\lambda_{2}$. Notation $|X^{\top}_{.j}\textbf{\textit{y}}|+\lambda_{2}\Big|_{j\in\{1,p\}}$ means values of $|X^{\top}_{.1}\textbf{\textit{y}}|+\lambda_{2}$ and
%$|X^{\top}_{.p}\textbf{\textit{y}}|+\lambda_{2}$. The value of $\lambda^{max}_{1}(\lambda_{2})$ is the maximal value of $|X^{\top}_{.2}\textbf{\textit{y}}|+2\lambda_{2}$, $|X^{\top}_{.3}\textbf{\textit{y}}|+2\lambda_{2},...,$ $|X^{\top}_{.p-1}\textbf{\textit{y}}|+2\lambda_{2}$, $|X^{\top}_{.1}\textbf{\textit{y}}|+\lambda_{2}$ and $|X^{\top}_{.p}\textbf{\textit{y}}|+\lambda_{2}$.}

On the other hand, when $\lambda_{1}\geq\lambda^{\rm{max}}_{1}(\lambda_{2})$, the above results all hold  and 0 is a solution of (\ref{eq:1}).

(\textit{ii}) If $\boldsymbol{\beta}^{*}(\lambda_{1},\lambda_{2})=\xi\textbf{1}_{p}$ with $\xi$ being a nonzero constant, $D\boldsymbol{\beta}^{*}(\lambda_{1},\lambda_{2})=0$.
In this case, $\lambda_{2}$ should satisfy that
\begin{align*}
&\textbf{\textit{y}}-\xi X\textbf{1}_{p}-\textbf{\textit{u}}^{*}(\lambda_{1},\lambda_{2})=0,\\
& X^{\top}\textbf{\textit{u}}^{*}(\lambda_{1},\lambda_{2})-D^{\top}\textbf{\textit{v}}^{*}(\lambda_{1},\lambda_{2})=\lambda_{1}\rm{sgn}(\xi)
\textbf{1}_{\textit{p}}
\end{align*}
 and  $\lambda_{2}\geq\|\textbf{\textit{v}}^{*}(\lambda_{1},\lambda_{2})\|_{\infty}$. So, we know that
 \begin{center}
 $X^{\top}(\textbf{\textit{y}}-\xi X\textbf{1}_{p})-D^{\top}\textbf{\textit{v}}^{*}(\lambda_{1},\lambda_{2})=\lambda_{1}\rm{sgn}(\xi)\textbf{1}_{\textit{p}}$,
 \end{center}
which leads to $D^{\top}\textbf{\textit{v}}^{*}(\lambda_{1},\lambda_{2})=X^{\top}\textbf{\textit{y}}-\xi X^{\top}X\textbf{1}_{p}-\lambda_{1}\rm{sgn}(\xi)\textbf{1}_{\textit{p}}$, i.e.,
\begin{align*}
v^{*}_{1}(\lambda_{1},\lambda_{2})&=X_{.1}^{\top}\textbf{\textit{y}}-\xi (X^{\top}X)_{1.}\textbf{1}_{p}-\lambda_{1}\rm{sgn}(\xi)\\
v^{*}_{2}(\lambda_{1},\lambda_{2})-v^{*}_{1}(\lambda_{1},\lambda_{2})&=X_{.2}^{\top}\textbf{\textit{y}}-\xi (X^{\top}X)_{2.}\textbf{1}_{\textit{p}}-\lambda_{1}\rm{sgn}(\xi)\\
&\vdots\\
v^{*}_{p-1}(\lambda_{1},\lambda_{2})-v^{*}_{p-2}(\lambda_{1},\lambda_{2})&=X_{.p-1}^{\top}\textbf{\textit{y}}-\xi (X^{\top}X)_{p-1.}\textbf{1}_{\textit{p}}-\lambda_{1}\rm{sgn}(\xi)\\
-v^{*}_{p}(\lambda_{1},\lambda_{2})&=X_{.p}^{\top}\textbf{\textit{y}}-\xi (X^{\top}X)_{p.}\textbf{1}_{\textit{p}}-\lambda_{1}\rm{sgn}(\xi).
\end{align*}
The closed-form expression of $v^{*}_{j}(\lambda_{1},\lambda_{2})$ is
\begin{align*}
&v^{*}_{j}(\lambda_{1},\lambda_{2})\\
&=
\begin{cases}
\sum\limits_{k=1}^{j}\{X_{.k}^{\top}\textbf{\textit{y}}-\xi (X^{\top}X)_{k.}\textbf{1}_{p}\}-j\lambda_{1}\rm{sgn}(\xi), &j\in\{1,2,\cdots,p-1\}\\
-X_{.p}^{\top}\textbf{\textit{y}}+\xi (X^{\top}X)_{p.}\textbf{1}_{p}+\lambda_{1}\rm{sgn}(\xi), &j=p
\end{cases}.
\end{align*}
In this case,
\begin{align*}
&\|\textbf{\textit{v}}^{*}(\lambda_{1},\lambda_{2})\|_{\infty}\\
&=
 \max\begin{pmatrix}
\underset{j\in\{1,2,\cdots,p-1\}}\max\Big|\sum\limits_{k=1}^{j}\{X^{\top}_{.k}\textbf{\textit{y}}-\xi(X^{\top}X)_{k.}\textbf{1}_{p}\}-j\lambda_{1}\rm{sgn}(\xi)\Big|,\\
|-X^{\top}_{.p}\textbf{\textit{y}}+\xi(X^{\top}X)_{p.}\textbf{1}_{p}+\lambda_{1}\rm{sgn}(\xi)|
\end{pmatrix}\\
&=\lambda^{max}_{2}(\lambda_{1}).
 \end{align*}
From these results, we know that $\lambda_{2}\geq\lambda^{max}_{2}(\lambda_{1})$.

On the other hand, when $\lambda_{2}\geq\lambda^{max}_{2}(\lambda_{1})$, the above results all hold  and there is a solution of (\ref{eq:1}) such that all of its elements are same.
\end{proof}
\begin{Remark}
Lemma 3.2 presents the lower bounds of $\lambda_{1}$ and $\lambda_{2}$ under the condition that one of them is fixed. As we all know, fused Lasso reduces to Lasso when $\lambda_{2}=0$, which is
\begin{center}
$\underset{\boldsymbol{\beta}\in \mathbb{R}^{p}}\min {\frac{1}{2}\|\textbf{\textit{y}}-X\boldsymbol{\beta}\|^{2}_{2}+\lambda\|\boldsymbol{\beta}\|_{1}}.$
\end{center}
The lower bound of $\lambda$ such that $\boldsymbol{\beta}^{*}(\lambda)=0$ is $\|X^{\top}\textbf{\textit{y}}\|_{\infty}$, which is calculated and used in many literatures, such as  Ghaoui et al. \cite{el2010safe}, Tibshirani et al. \cite{tibshirani2012strong}, Wang et al. \cite{wang2015fused} and Ndiaye et al. \cite{ndiaye2017gap}.  When $\lambda_{2}=0$, $\lambda^{max}_{1}(\lambda_{2})=\|X^{\top}\textbf{\textit{y}}\|_{\infty}$ and it is same with the lower bound of $\lambda$.  Similarly, the
fused Lasso reduces to the following model when  $\lambda_{1}=0$, which is
\begin{center}
$\underset{\boldsymbol{\beta}\in \mathbb{R}^{p}}\min {\frac{1}{2}\|\textbf{\textit{y}}-X\boldsymbol{\beta}\|^{2}_{2}+\lambda\|D\boldsymbol{\beta}\|_{1}}.$
\end{center}
 In Wang et al. \cite{wang2015fused}, they showed that the lower bound of $\lambda$ such that $\boldsymbol{\beta}^{*}(\lambda)=\xi\textbf{1}_{p}$ is the value of $\lambda^{max}_{2}(\lambda_{1})$ in Lemma 3.2 when $\lambda_{1}=0$.
\end{Remark}

From this lemma, we only need to estimate the solution of fused Lasso when $\lambda_{1}\in(0,\lambda^{max}_{1}(\lambda_{2}))$ with fixed $\lambda_{2}$ or $\lambda_{2}\in(0,\lambda^{max}_{2}(\lambda_{1}))$ with fixed $\lambda_{1}$. In the following part, we fix the value of $\lambda_{2}$. In the next lemma, we estimate the dual solution $\textbf{\textit{u}}^{*}(\lambda_{1},\lambda_{2})$ with the help of the variational inequality.
\begin{Lemma}
For any fixed $\lambda_{2}>0$ and $\tilde{\lambda}_{1}\in(0,\lambda^{max}_{1}(\lambda_{2})]$\textcolor{black}{, suppose } $\textbf{\textit{u}}^{*}(\tilde{\lambda}_{1},\lambda_{2})$ is known and $\lambda_{1}\in(0,\tilde{\lambda}_{1})$. The dual solution satisfies that
\begin{center}
$\textbf{\textit{u}}^{*}(\lambda_{1},\lambda_{2})\in\Omega:=\{\textbf{\textit{u}}:\|\textbf{\textit{u}}-\textbf{c}\|_{2}\leq r\}$
\end{center}
with $\textbf{c}=\frac{1}{2}\left[(1+\frac{\lambda_{1}}{\tilde{\lambda}_{1}})\textbf{\textit{y}}-
\frac{\lambda_{1}}{\tilde{\lambda}_{1}}\textbf{\textit{u}}^{*}(\tilde{\lambda}_{1},\lambda_{2})\right]$ and $r=\|\textbf{c}\|_{2}$.
\end{Lemma}
\begin{proof}
For ease of expression, we denote
\begin{equation*}
\boldsymbol{\theta}=\begin{pmatrix}
\textbf{\textit{u}}\\
\textbf{\textit{v}}
\end{pmatrix}\in \mathbb{R}^{n+p-1},
\tilde{\textbf{\textit{y}}}=\begin{pmatrix}
\textbf{\textit{y}}\\
0
\end{pmatrix}\in \mathbb{R}^{n+p-1},
A=\begin{pmatrix}I_{n},0\\
0,0
\end{pmatrix}\in \mathbb{R}^{(n+p-1)\times(n+p-1)},
\end{equation*}
\begin{equation*}
B=\begin{pmatrix}X^{\top},-D^{\top}
\end{pmatrix}\in \mathbb{R}^{p\times(n+p-1)},
C=\begin{pmatrix}0,I_{p-1}
\end{pmatrix}\in \mathbb{R}^{(p-1)\times(n+p-1)}.
\end{equation*}
Then, the dual problem (\ref{eq:3}) is equivalent to
\begin{equation*}
\begin{split} \underset{\boldsymbol{\theta}\in \mathbb{R}^{n+p-1}}\min  &\frac{1}{2}\boldsymbol{\theta}^{\top}A\boldsymbol{\theta}-\langle \boldsymbol{\theta},\tilde{\textbf{\textit{y}}}\rangle\\
 s.t. \quad &\|B\boldsymbol{\theta}\|_{\infty}\leq \lambda_{1}, \|C\boldsymbol{\theta}\|_{\infty}\leq\lambda_{2}.
\end{split}
\end{equation*}
Denote $G(\boldsymbol{\theta})=\frac{1}{2}\boldsymbol{\theta}^{\top}A\boldsymbol{\theta}-\langle \boldsymbol{\theta},\tilde{\textbf{\textit{y}}}\rangle$. Then $\nabla G(\boldsymbol{\theta})=A\boldsymbol{\theta}-\tilde{\textbf{\textit{y}}}$. According to the variational inequality, we know that
\begin{center}
$\left\langle\nabla G(\boldsymbol{\theta}^{*}(\lambda_{1},\lambda_{2})),\boldsymbol{\theta}-\boldsymbol{\theta}^{*}(\lambda_{1},\lambda_{2})\right\rangle\geq0$
\end{center}
holds for any $\boldsymbol{\theta}\in\mathcal{F}=\{\boldsymbol{\theta}|\|B\boldsymbol{\theta}\|_{\infty}\leq \lambda_{1},\|C\boldsymbol{\theta}\|_{\infty}\leq\lambda_{2}\}$. Hence, we have
\begin{equation}\label{eq:7}
\left\langle A\boldsymbol{\theta}^{*}\left(\tilde{\lambda}_{1},\lambda_{2}\right)-\tilde{\textbf{\textit{y}}},\boldsymbol{\theta}^{*}(\lambda_{1},\lambda_{2}) -\boldsymbol{\theta}^{*}\left(\tilde{\lambda}_{1},\lambda_{2}\right)\right\rangle\geq0
\end{equation}
and
\begin{equation}\label{eq:8}
\left\langle A\boldsymbol{\theta}^{*}(\lambda_{1},\lambda_{2})-\tilde{\textbf{\textit{y}}},
\frac{\lambda_{1}}{\tilde{\lambda}_{1}}\boldsymbol{\theta}^{*}\left(\tilde{\lambda}_{1},\lambda_{2}\right) -\boldsymbol{\theta}^{*}(\lambda_{1},\lambda_{2})\right\rangle\geq0.
\end{equation}
By adding (\ref{eq:8}) with (\ref{eq:7}) multiplying with $\frac{\lambda_{1}}{\tilde{\lambda}_{1}}$, we get
$$\left\langle\frac{\lambda_{1}}{\tilde{\lambda}_{1}}A\boldsymbol{\theta}^{*}(\tilde{\lambda}_{1},\lambda_{2})
+A\boldsymbol{\theta}^{*}(\lambda_{1},\lambda_{2})-(1+\frac{\lambda_{1}}{\tilde{\lambda}_{1}})\tilde{\textbf{\textit{y}}},
(\frac{\lambda_{1}}{\tilde{\lambda}_{1}}-1)\boldsymbol{\theta}^{*}(\lambda_{1},\lambda_{2})\right\rangle\geq0,$$
which leads to
\begin{align*}
&\left\langle\frac{\tilde{\lambda}_{1}}{\lambda_{2}}A\boldsymbol{\theta}^{*}\left(\tilde{\lambda}_{1},\lambda_{2}\right), \left(\frac{\lambda_{1}}{\tilde{\lambda}_{1}}-1\right)\boldsymbol{\theta}^{*}(\lambda_{1},\lambda_{2})\right\rangle\\
&\geq\left\langle A\boldsymbol{\theta}^{*}(\lambda_{1},\lambda_{2}),
\left(1-\frac{\lambda_{1}}{\tilde{\lambda}_{1}}\right)\boldsymbol{\theta}^{*}(\lambda_{1},\lambda_{2})\right\rangle\\
&\quad +\left(1+\frac{\lambda_{1}}{\tilde{\lambda}_{1}}\right)\left(\frac{\lambda_{1}}{\tilde{\lambda}_{1}}-1\right)
\left\langle \tilde{\textbf{\textit{y}}}, \boldsymbol{\theta}^{*}(\lambda_{1},\lambda_{2})\right\rangle
\end{align*}
and
\begin{align*}
&\left(\frac{\lambda_{1}}{\tilde{\lambda}_{1}}-1\right)\frac{\lambda_{1}}{\tilde{\lambda}_{1}}\left\langle \boldsymbol{\theta}^{*}\left(\lambda_{1},\tilde{\lambda}_{2}\right), A\boldsymbol{\theta}^{*}(\lambda_{1},\lambda_{2})\right\rangle\\
&\geq \left(1-\frac{\lambda_{1}}{\tilde{\lambda}_{1}}\right)\left\langle A\boldsymbol{\theta}^{*}(\lambda_{1},\lambda_{2}),\boldsymbol{\theta}^{*}(\lambda_{1},\lambda_{2})\right\rangle\\
&\quad +\left(1+\frac{\lambda_{1}}{\tilde{\lambda}_{1}}\right)\left(\frac{\lambda_{1}}{\tilde{\lambda}_{1}}-1\right)
\left\langle \tilde{\textbf{\textit{y}}}, \boldsymbol{\theta}^{*}(\lambda_{1},\lambda_{2})\right\rangle.
\end{align*}
Dividing the last inequality by $1-\frac{\lambda_{1}}{\tilde{\lambda}_{1}}$, we get
\begin{equation}\label{eq:9}
\begin{aligned}
&-\frac{\lambda_{1}}{\tilde{\lambda}_{1}}\left\langle \boldsymbol{\theta}^{*}\left(\tilde{\lambda}_{1},\lambda_{2}\right), A\boldsymbol{\theta}^{*}(\lambda_{1},\lambda_{2})\right\rangle\\
&~~\geq \left\langle A\boldsymbol{\theta}^{*}(\lambda_{1},\lambda_{2}),\boldsymbol{\theta}^{*}(\lambda_{1},\lambda_{2})\right\rangle-\left(1+\frac{\lambda_{1}}{\tilde{\lambda}_{1}}\right)\left\langle \tilde{\textbf{\textit{y}}}, \boldsymbol{\theta}^{*}(\lambda_{1},\lambda_{2})\right\rangle.
\end{aligned}
\end{equation}
Because of
\begin{equation*}
A^{\top}\boldsymbol{\theta}^{*}(\lambda_{1},\lambda_{2})
=\begin{pmatrix}
\textbf{\textit{u}}^{*}(\lambda_{1},\lambda_{2})\\
0
\end{pmatrix}
\end{equation*}
and $\langle \tilde{\textbf{\textit{y}}}, \boldsymbol{\theta}^{*}(\lambda_{1},\lambda_{2})\rangle=\langle\textbf{\textit{y}},\textbf{\textit{u}}^{*}(\lambda_{1},\lambda_{2})\rangle$, the equation (\ref{eq:9}) can be written as
\begin{align*}
&-\frac{\lambda_{1}}{\tilde{\lambda}_{1}}\left\langle \textbf{\textit{u}}^{*}\left(\tilde{\lambda}_{1},\lambda_{2}\right), \textbf{\textit{u}}^{*}(\lambda_{1},\lambda_{2})\right\rangle\\
&\quad \geq\|\textbf{\textit{u}}^{*}(\lambda_{1},\lambda_{2})\|^{2}_{2}-
\left(1+\frac{\lambda_{1}}{\tilde{\lambda}_{1}}\right)\left\langle \textbf{\textit{y}},\textbf{\textit{u}}^{*}(\lambda_{1},\lambda_{2})\right\rangle,
\end{align*}
which means
\begin{align*}
&\|\textbf{\textit{u}}^{*}(\lambda_{1},\lambda_{2})\|^{2}_{2}-\left(1+\frac{\lambda_{1}}{\tilde{\lambda}_{1}}\right)\langle \textbf{\textit{y}},\textbf{\textit{u}}^{*}(\lambda_{1},\lambda_{2})\rangle\\
&\quad\quad +\frac{\lambda_{1}}{\tilde{\lambda}_{1}}\left\langle \textbf{\textit{u}}^{*}(\lambda_{1},\lambda_{2}), \textbf{\textit{u}}^{*}\left(\tilde{\lambda}_{1},\lambda_{2}\right)\right\rangle\leq0.
\end{align*}
Therefore,
\begin{align*}
&\left\|\textbf{\textit{u}}^{*}(\lambda_{1},\lambda_{2})-\frac{1}{2}\left(1+\frac{\lambda_{1}}{\tilde{\lambda}_{1}}\right)
\textbf{\textit{y}}+
\frac{1}{2}\frac{\lambda_{1}}{\tilde{\lambda}_{1}}\textbf{\textit{u}}^{*}\left(\tilde{\lambda}_{1},\lambda_{2}\right)\right\|^{2}_{2}\\
&\quad \leq \frac{1}{4}\left(1+\frac{\lambda_{1}}{\tilde{\lambda}_{1}}\right)^{2}\|\textbf{\textit{y}}\|^{2}_{2}
+\frac{1}{4}\left(\frac{\lambda_{1}}{\tilde{\lambda}_{1}}\right)^{2}\left\|\textbf{\textit{u}}^{*}(\tilde{\lambda}_{1},\lambda_{2})\right\|^{2}_{2}\\
&\quad\quad\quad\quad\quad\quad -\frac{1}{2}\left(1+\frac{\lambda_{1}}{\tilde{\lambda}_{1}}\right)\frac{\lambda_{1}}{\tilde{\lambda}_{1}}
\left\langle \textbf{\textit{y}},\textbf{\textit{u}}^{*}\left(\tilde{\lambda}_{1},\lambda_{2}\right)\right\rangle,
\end{align*}
which leads to $\|\textbf{\textit{u}}^{*}(\lambda_{1},\lambda_{2})-\textbf{\textit{c}}\|^{2}_{2}\leq r^{2}$ and the result is proved.
\end{proof}
\begin{Corollary}
In a special case that $\tilde{\lambda}_{1}=\lambda^{\max}_{1}(\lambda_{2})$,
\begin{center}
$\textbf{\textit{u}}^{*}\left(\tilde{\lambda}_{1},\lambda_{2}\right)
=\textbf{\textit{u}}^{*}\left(\lambda^{\max}_{1}(\lambda_{2}),\lambda_{2}\right)=\textbf{\textit{y}}$.
\end{center}
Here, $\textbf{c}=\frac{1}{2}\textbf{\textit{y}}$ and $\Omega=\{\textbf{\textit{u}}: \|\textbf{\textit{u}}-\frac{1}{2}\textbf{\textit{y}}\|_{2}\leq\frac{1}{2}\|\textbf{\textit{y}}\|_{2}\}$.
\end{Corollary}

\subsection{Safe feature identification rule for fused Lasso}
According to the duality theory in Section 3.1 and the estimation of the dual solution $\textbf{\textit{u}}^{*}(\lambda_{1},\lambda_{2})$ in  Section 3.2, we now present the safe feature identification rule for fused Lasso. By Cauchy inequality, the estimation in Lemma 3.4 implies that
$$|X_{.j}^{\top}\textbf{\textit{u}}^{*}(\lambda_{1},\lambda_{2})|\leq |X^{\top}_{.j}\textbf{c}|+r\|X_{.j}\|_{2}.$$
\begin{Theorem}
For any fixed $\lambda_{2}>0$ and $\tilde{\lambda}_{1}\in(0,\lambda^{max}_{1}(\lambda_{2})]$, \textcolor{black}{let} $\lambda_{1}\in(0,\tilde{\lambda}_{1})$. Then $\beta^{*}_{j}(\lambda_{1},\lambda_{2})=0$ if
\begin{equation*}
\begin{cases}
|X^{\top}_{.j}\textbf{c}|+r\|X_{.j}\|_{2}<\lambda_{1}-\lambda_{2},&~~ j\in\{1, p\},\\
|X^{\top}_{.j}\textbf{c}|+r\|X_{.j}\|_{2}<\lambda_{1}-2\lambda_{2},&~~j\in\{2, \cdots, p-1\}.
\end{cases}
\end{equation*}
\end{Theorem}
\begin{proof}
Based on $(i)$ in Lemma 3.1, we know $\beta^{*}_{j}(\lambda_{1},\lambda_{2})=0$, if
\begin{center} $\Big|X_{.j}^{\top}\textbf{\textit{u}}^{*}(\lambda_{1},\lambda_{2})-D_{.j}^{\top}\textbf{\textit{v}}^{*}(\lambda_{1},\lambda_{2})\Big|<\lambda_{1}$.
\end{center}
We show the detailed results of this conclusion.
\begin{equation*}
\begin{cases}
\beta^{*}_{1}(\lambda_{1},\lambda_{2})=0,\\
\quad\quad\quad\quad\quad\quad |X_{.1}^{\top}\textbf{\textit{u}}^{*}(\lambda_{1},\lambda_{2})-v_{1}^{*}(\lambda_{1},\lambda_{2})|<\lambda_{1},\\
\beta^{*}_{j}(\lambda_{1},\lambda_{2})=0, \\ \quad\quad\quad\quad\quad\quad|X_{.j}^{\top}\textbf{\textit{u}}^{*}(\lambda_{1},\lambda_{2})-v_{j}^{*}(\lambda_{1},\lambda_{2})+v_{j-1}^{*}(\lambda_{1},\lambda_{2})|<\lambda_{1},\\
\quad\quad\quad\quad\quad\quad j\in\{2,\cdots,p-1\},\\
\beta_{p}^{*}(\lambda_{1},\lambda_{2})=0,\\ \quad\quad\quad\quad\quad\quad|X_{.p}^{\top}\textbf{\textit{u}}^{*}(\lambda_{1},\lambda_{2})+v_{p-1}^{*}(\lambda_{1},\lambda_{2})|<\lambda_{1}.
\end{cases}
\end{equation*}
According to Lemma 3.3 and Cauchy inequality, we know that
\begin{align*}
&|X_{.1}^{\top}\textbf{\textit{u}}^{*}(\lambda_{1},\lambda_{2})-D_{.1}^{\top}\textbf{\textit{v}}^{*}(\lambda_{1},\lambda_{2})|\\
&\quad =|X_{.1}^{\top}\textbf{\textit{u}}^{*}(\lambda_{1},\lambda_{2})-v_{1}^{*}(\lambda_{1},\lambda_{2})|\\
&\quad \leq \underset{\textbf{\textit{u}}\in\Omega, |v_{1}|\leq\lambda_{2}}\max|X_{.1}^{\top}\textbf{\textit{u}}-v_{1}|\\
&\quad \leq\underset{\textbf{\textit{u}}\in\Omega}\max|X_{.1}^{\top}\textbf{\textit{u}}|+ \underset{|v_{1}|\leq\lambda_{2}}\max|-v_{1}|\\
&\quad \leq \underset{\|\delta\|_{2}\leq r}\max|X_{.1}^{\top}(\textbf{c}+\delta)|+\lambda_{2}\\
&\quad =|X_{.1}^{\top}\textbf{c}|+r\|X_{.1}\|_{2}+\lambda_{2}.
\end{align*}
So, $\beta^{*}_{1}(\lambda_{1},\lambda_{2})=0$ if $|X_{.1}^{\top}\textbf{c}|+r\|X_{.1}\|_{2}+\lambda_{2}<\lambda_{1}$. In the same way, for $j\in\{2,\cdots,p-1\}$, we get that
\begin{align*}
&|X_{.j}^{\top}\textbf{\textit{u}}^{*}(\lambda_{1},\lambda_{2})-D_{.j}^{\top}\textbf{\textit{v}}^{*}(\lambda_{1},\lambda_{2})|\\
&\quad=|X_{.j}^{\top}\textbf{\textit{u}}^{*}(\lambda_{1},\lambda_{2})-v_{j}^{*}(\lambda_{1},\lambda_{2})+v_{j-1}^{*}(\lambda_{1},\lambda_{2})|\\
&\quad \leq|X_{.j}^{\top}\textbf{c}|+r\|X_{.j}\|_{2}+2\lambda_{2}.
\end{align*}
And
\begin{align*}
&|X_{.p}^{\top}\textbf{\textit{u}}^{*}(\lambda_{1},\lambda_{2})-D_{.p}^{\top}\textbf{\textit{v}}^{*}(\lambda_{1},\lambda_{2})|\\
&\quad =|X_{.p}^{\top}\textbf{\textit{u}}^{*}(\lambda_{1},\lambda_{2})+v_{p-1}^{*}(\lambda_{1},\lambda_{2})|\\
&\quad \leq|X_{.p}^{\top}\textbf{c}|+r\|X_{.p}\|_{2}+\lambda_{2}.
\end{align*}
Therefore, the desired result can be concluded.
\end{proof}

Theorem 3.1 can be used to eliminate inactive features under different $\lambda_{1}$ and $\lambda_{2}$. Unlike the result in Lemma 3.1, Theorem 3.1 needs the known solution $\textbf{\textit{u}}^{*}(\tilde{\lambda}_{1},\lambda_{2})$ and this result can be applied to accelerate the computation of the solution path of fused Lasso. Next, we show the result that can be used to identify successively same elements of the solution of fused Lasso.
\begin{Theorem}
For any fixed $\lambda_{2}>0$ and $\tilde{\lambda}_{1}\in(0,\lambda^{max}_{1}(\lambda_{2})]$, \textcolor{black}{let} $\lambda_{1}\in(0,\tilde{\lambda}_{1})$.
\begin{enumerate}[1)]
\item {For any $j\in\{1,p-1\}$, $\beta^{*}_{j}(\lambda_{1},\lambda_{2})=\beta^{*}_{j+1}(\lambda_{1},\lambda_{2})$ holds if
\begin{center}
$|X^{\top}_{.j}\textbf{c}|+r\|X_{.j}\|_{2}<\lambda_{2}-\lambda_{1}$.
\end{center}}
\item {For any $j\in\{2,3,\cdots,p-2\}$, $\beta^{*}_{j}(\lambda_{1},\lambda_{2})=\beta^{*}_{j+1}(\lambda_{1},\lambda_{2})$ holds if
\begin{center}
$|X^{\top}_{.j}\textbf{c}|+r\|X_{.j}\|_{2}<2\lambda_{2}+\lambda_{1}$.
\end{center}}
\end{enumerate}
\end{Theorem}
\begin{proof}
For any $j\in\{1,2,\cdots,p-1\}$,
\begin{align*}
\begin{split}
|v^{*}_{j}(\lambda_{1},\lambda_{2})|&\leq\max|v_{j}|\\
&\quad s.t.~\|X^{\top}\textbf{\textit{u}}^{*}(\lambda_{1},\lambda_{2})-D^{\top}\textbf{\textit{v}}\|_{\infty}\leq\lambda_{1}\\
&\quad~~~~~~\|\textbf{\textit{v}}\|_{\infty}\leq\lambda_{2}.
\end{split}
\end{align*}

For $j=1$, it holds that
\begin{align*}
\begin{split}
|v^{*}_{1}(\lambda_{1},\lambda_{2})|&\leq\max|v_{1}|\\
&\quad s.t.~|X^{\top}_{.1}\textbf{\textit{u}}^{*}(\lambda_{1},\lambda_{2})-v_{1}|\leq\lambda_{1}\\
&\quad~~~~~|v_{1}|\leq\lambda_{2}
\end{split}\\
&\leq \min\left\{|X^{\top}_{.1}\textbf{\textit{u}}^{*}(\lambda_{1},\lambda_{2})|+\lambda_{1},\lambda_{2}\right\}.
\end{align*}
The constraint $|X^{\top}_{.1}\textbf{\textit{u}}^{*}(\lambda_{1},\lambda_{2})-v_{1}|\leq\lambda_{1}$ leads to
$$X^{\top}_{.1}\textbf{\textit{u}}^{*}(\lambda_{1},\lambda_{2})-\lambda_{1}\leq v^{*}_{1}(\lambda_{1},\lambda_{2})\leq X^{\top}_{.1}\textbf{\textit{u}}^{*}(\lambda_{1},\lambda_{2})+\lambda_{1},$$
which implies
\begin{align*}
|v^{*}_{1}(\lambda_{1},\lambda_{2})|&\leq
|X^{\top}_{.1}\textbf{\textit{u}}^{*}(\lambda_{1},\lambda_{2})+\lambda_{1}|\\
&\leq |X^{\top}_{.1}\textbf{c}|+r\|X_{.1}\|_{2}+\lambda_{1}.
\end{align*}
So, $|v^{*}_{1}(\lambda_{1},\lambda_{2})|\leq\max\{ |X^{\top}_{.1}\textbf{c}|+r\|X_{.1}\|_{2}+\lambda_{1},\lambda_{2}\}$.

For $j=p-1$, same as $j=1$, it holds that
\begin{align*}
\begin{split}
|v^{*}_{p-1}(\lambda_{1},\lambda_{2})|&\leq\max|v_{p-1}|\\
&\quad s.t.~|X^{\top}_{.p-1}\textbf{\textit{u}}^{*}(\lambda_{1},\lambda_{2})+v_{p-1}|\leq\lambda_{1}\\
&\quad~~~~~~|v_{p-1}|\leq\lambda_{2}
\end{split}\\
&\leq \min\left\{|X^{\top}_{.p-1}\textbf{\textit{u}}^{*}(\lambda_{1},\lambda_{2})|+\lambda_{1},\lambda_{2}\right\}\\
&\leq \min\left\{|X^{\top}_{.p-1}\textbf{c}|+r\|X_{.p-1}\|_{2}+\lambda_{1},\lambda_{2}\right\}.
\end{align*}

When $j\in\{1,p-1\}$, according to the result (ii) in Lemma 3.2,  $\beta^{*}_{j}(\lambda_{1},\lambda_{2})=\beta^{*}_{j+1}(\lambda_{1},\lambda_{2})$ if $$\min\left\{|X^{\top}_{.j}\textbf{c}|+r\|X_{.j}\|_{2}+\lambda_{1},\lambda_{2}\right\}<\lambda_{2},$$
which means $$|X^{\top}_{.j}\textbf{c}|+r\|X_{.j}\|_{2}<\lambda_{2}-\lambda_{1}.$$

When $j\in\{2,\cdots,p-2\}$,
\begin{align*}
\begin{split}
|v^{*}_{j}(\lambda_{1},\lambda_{2})|&\leq \max|v_{j}|\\
&\quad s.t.~|X^{\top}_{.j}\textbf{\textit{u}}^{*}(\lambda_{1},\lambda_{2})-v_{j}+v^{*}_{j-1}(\lambda_{1},\lambda_{2})|\leq\lambda_{1}\\
&\quad ~~~~~~|X^{\top}_{.j+1}\textbf{\textit{u}}^{*}(\lambda_{1},\lambda_{2})-v^{*}_{j+1}(\lambda_{1},\lambda_{2})+v_{j}|\leq\lambda_{1}\\
&\quad ~~~~~~|v_{j}|\leq\lambda_{2}
\end{split}\\
\begin{split}&\leq
\max|v_{j}|\\
&\quad s.t.~|X^{\top}_{.j}\textbf{\textit{u}}^{*}(\lambda_{1},\lambda_{2})-v_{j}+v^{*}_{j-1}(\lambda_{1},\lambda_{2})|\leq\lambda_{1}\\
&\quad ~~~~~~|v_{j}|\leq\lambda_{2}
\end{split}.
\end{align*}
The constraint $|X^{\top}_{.j}\textbf{\textit{u}}^{*}(\lambda_{1},\lambda_{2})-v_{j}+v^{*}_{j-1}(\lambda_{1},\lambda_{2})|\leq\lambda_{1}$ leads to
$v^{*}_{j}(\lambda_{1},\lambda_{2})\geq X^{\top}_{.j}\textbf{\textit{u}}^{*}(\lambda_{1},\lambda_{2})+v^{*}_{j-1}(\lambda_{1},\lambda_{2})-\lambda_{1}$
and
$v^{*}_{j}(\lambda_{1},\lambda_{2})\leq
X^{\top}_{.j}\textbf{\textit{u}}^{*}(\lambda_{1},\lambda_{2})+v^{*}_{j-1}(\lambda_{1},\lambda_{2})+\lambda_{1},$
which implies that
\begin{align*}
|v^{*}_{j}(\lambda_{1},\lambda_{2})|&\leq\tau_{j}:=\max\{
|X^{\top}_{.j}\textbf{\textit{u}}^{*}(\lambda_{1},\lambda_{2})+v^{*}_{j-1}(\lambda_{1},\lambda_{2})+\lambda_{1}|,\\
&\quad\quad\quad\quad\quad\quad
|X^{\top}_{.j}\textbf{\textit{u}}^{*}(\lambda_{1},\lambda_{2})+v^{*}_{j-1}(\lambda_{1},\lambda_{2})-\lambda_{1}|\}.
\end{align*}
Combining the constraint $|v_{j}^{*}(\lambda_{1},\lambda_{2})|\leq\lambda_{2}$, then
\begin{align*}
|v^{*}_{j}(\lambda_{1},\lambda_{2})|\leq \min\left\{\tau_{j},\lambda_{2}\right\}.
\end{align*}
To make sure $|v^{*}_{j}(\lambda_{1},\lambda_{2})|<\lambda_{2}$, the condition $\tau_{j}<\lambda_{2}$ needs to be required, which means\\
\indent(a) $|X^{\top}_{.j}\textbf{\textit{u}}^{*}(\lambda_{1},\lambda_{2})+v^{*}_{j-1}(\lambda_{1},\lambda_{2})+\lambda_{1}|<\lambda_{2}$ \\or \\
\indent(b) $|X^{\top}_{.j}\textbf{\textit{u}}^{*}(\lambda_{1},\lambda_{2})+v^{*}_{j-1}(\lambda_{1},\lambda_{2})-\lambda_{1}|<\lambda_{2}$.\\
\noindent The condition (a) means
$$X^{\top}_{.j}\textbf{\textit{u}}^{*}(\lambda_{1},\lambda_{2})>-\lambda_{2}-v^{*}_{j-1}(\lambda_{1},\lambda_{2})-\lambda_{1}\geq -2\lambda_{2}-\lambda_{1}$$
and
$$X^{\top}_{.j}\textbf{\textit{u}}^{*}(\lambda_{1},\lambda_{2})<\lambda_{2}-v^{*}_{j-1}(\lambda_{1},\lambda_{2})-\lambda_{1}\leq 2\lambda_{2}-\lambda_{1},$$
which implies
$$|X^{\top}_{.j}\textbf{\textit{u}}^{*}(\lambda_{1},\lambda_{2})|< \max\{2\lambda_{2}+\lambda_{1},|2\lambda_{2}-\lambda_{1}|\}=2\lambda_{2}+\lambda_{1}.$$
\noindent The condition (b) means
$$X^{\top}_{.j}\textbf{\textit{u}}^{*}(\lambda_{1},\lambda_{2})>-\lambda_{2}-v^{*}_{j-1}(\lambda_{1},\lambda_{2})+\lambda_{1}\geq -2\lambda_{2}+\lambda_{1}$$
and
$$X^{\top}_{.j}\textbf{\textit{u}}^{*}(\lambda_{1},\lambda_{2})<\lambda_{2}-v^{*}_{j-1}(\lambda_{1},\lambda_{2})+\lambda_{1}\leq 2\lambda_{2}+\lambda_{1},$$
which implies
$$|X^{\top}_{.j}\textbf{\textit{u}}^{*}(\lambda_{1},\lambda_{2})|< \max\{2\lambda_{2}+\lambda_{1},|2\lambda_{2}-\lambda_{1}|\}=2\lambda_{2}+\lambda_{1}.$$\\
\noindent These analysis implies that $|v^{*}_{j}(\lambda_{1},\lambda_{2})|<\lambda_{2}$ holds if  $$|X^{\top}_{.j}\textbf{\textit{u}}^{*}(\lambda_{1},\lambda_{2})|< 2\lambda_{2}+\lambda_{1},$$
which can be guaranteed by
$|X^{\top}_{.j}\textbf{c}|+r\|X_{.j}\|_{2}<2\lambda_{2}+\lambda_{1}.$
%\begin{align*}
%\begin{split}
%|v^{*}_{j}(\lambda_{1},\lambda_{2})|&\geq \min|v_{j}|\\
%&\quad s.t.~|X^{\top}_{.j}\textbf{\textit{u}}^{*}(\lambda_{1},\lambda_{2})-v_{j}+v^{*}_{j-1}(\lambda_{1},\lambda_{2})|\leq\lambda_{1}\\
%&\quad ~~~~~|X^{\top}_{.j+1}\textbf{\textit{u}}^{*}(\lambda_{1},\lambda_{2})-v^{*}_{j+1}(\lambda_{1},\lambda_{2})+v_{j}|\leq\lambda_{1}\\
%&\quad ~~~~~|v_{j}|\leq\lambda_{2}
%\end{split}\\
%\begin{split}
%&\geq
%\min|v_{j}|\\
%&\quad s.t.~|X^{\top}_{.j}\textbf{\textit{u}}^{*}(\lambda_{1},\lambda_{2})-v_{j}+v^{*}_{j-1}(\lambda_{1},\lambda_{2})|\leq\lambda_{1}\\
%&\quad ~~~~~|v_{j}|\leq\lambda_{2}
%\end{split}\\
%&\geq |X^{\top}_{.j}\textbf{\textit{u}}^{*}(\lambda_{1},\lambda_{2})|-|v^{*}_{j-1}(\lambda_{1},\lambda_{2})|-\lambda_{1}
%\end{align*}

%For $j=2$, when $|v^{*}_{1}(\lambda_{1},\lambda_{2})|=\lambda_{2}$, the result $|v^{*}_{2}(\lambda_{1},\lambda_{2})|<2\lambda_{2}$ always hold. To make sure $v^{*}_{2}(\lambda_{1},\lambda_{2})>0$,  we need
%\begin{center}
%$|X^{\top}_{.j}\textbf{\textit{u}}^{*}(\lambda_{1},\lambda_{2})|>|v^{*}_{j-1}(\lambda_{1},\lambda_{2})|+\lambda_{1}$,
%\end{center}
%which can be guaranteed by
%$$|X^{\top}_{.j}\textbf{c}|-r\|X_{.j}\|_{2}>\lambda_{1}+\lambda_{2}.$$
%
%For $j\in\{3,4,\cdots,p-2\}$, when $|v^{*}_{j-1}(\lambda_{1},\lambda_{2})|\in\{0,2\lambda_{2}\}$, the result $|v^{*}_{j}(\lambda_{1},\lambda_{2})|<2\lambda_{2}$ always hold. Same as the analysis for $j=2$, $v^{*}_{j}(\lambda_{1},\lambda_{2})>0$ can be guaranteed by
%$$|X^{\top}_{.j}\textbf{c}|-r\|X_{.j}\|_{2}>\lambda_{1}+2\lambda_{2}.$$
\end{proof}

Replace $\textbf{c}=\frac{1}{2}\textbf{\textit{y}}$ and $r=\frac{1}{2}\|\textbf{\textit{y}}\|_{2}$ in Theorem 3.1 and Theorem 3.2, there are some closed-form results.
\begin{Corollary}
Let $\lambda_{2}>0$  and  $\lambda_{1}\in(0,\lambda^{max}_{1}(\lambda_{2}))$.

For any $j\in\{1, \cdots, p\}$, $\beta^{*}_{j}(\lambda_{1},\lambda_{2})=0$ if
\begin{equation*}
\begin{cases}
\|\textbf{\textit{y}}\|_{2}\|X_{.j}\|_{2}+|X^{\top}_{.j}\textbf{\textit{y}}|<2(\lambda_{1}-\lambda_{2}),&~~ j\in\{1, p\},\\
\|\textbf{\textit{y}}\|_{2}\|X_{.j}\|_{2}+|X^{\top}_{.j}\textbf{\textit{y}}|<2(\lambda_{1}-2\lambda_{2}),&~~j\in\{2, \cdots, p-1\}.
\end{cases}
\end{equation*}
The result $\beta^{*}_{j}(\lambda_{1},\lambda_{2})=\beta^{*}_{j+1}(\lambda_{1},\lambda_{2})$ if
\begin{equation*}
\begin{cases}
\|\textbf{\textit{y}}\|_{2}\|X_{.j}\|_{2}+|X^{\top}_{.j}\textbf{\textit{y}}|<2(\lambda_{2}-\lambda_{1}),&~~ j\in\{1, p-1\},\\
\|\textbf{\textit{y}}\|_{2}\|X_{.j}\|_{2}+|X^{\top}_{.j}\textbf{\textit{y}}|<2(2\lambda_{2}+\lambda_{1}),&~~j\in\{2, \cdots, p-2\}.
\end{cases}
\end{equation*}
\end{Corollary}
\begin{Remark}
In this paper, the safe feature identification rule is composed of Theorem 3.1 and Theorem 3.2, where Theorem 3.1 eliminates features with 0 coefficient and Theorem 3.2 identifies the features with same coefficients.  The main computations of our screening rule are the values of $|X^{\top}_{.j}\textbf{c}|$ and $\|X_{.j}\|_{2}$, where $j\in\{1,2,\cdots,p\}$. This leads to the time complexity of our rule is $O(np)$. Here, we present the detailed process that uses Theorem 3.1 and Theorem 3.2 to accelerate the computation of fused Lasso. For any fixed $\lambda_{2}>0$, suppose we have a sequence values of $\lambda_{1}$ such that $\lambda^{max}_{1}(\lambda_{2})=\lambda^{(1)}_{1}>\lambda^{(2)}_{1}>\cdots>\lambda^{(k)}_{1}>0.$
 Because $\boldsymbol{\beta}^{*}(\lambda^{(1)}_{1},\lambda_{2})=0$ and $\textbf{\textit{u}}^{*}(\lambda^{(1)}_{1},\lambda_{2})=\textbf{\textit{y}}$, the inactive features can be eliminated by setting $\tilde{\lambda}_{1}=\lambda^{(1)}_{1}$ and $\lambda_{1}=\lambda^{(2)}_{1}$ in Theorem 3.1 and Theorem 3.2. Under the help of our screening rule, $\boldsymbol{\beta}^{*}(\lambda^{(2)}_{1},\lambda_{2})$ and $\textbf{\textit{u}}^{*}(\lambda^{(2)}_{1},\lambda_{2})$ can be solved on the reduced data set. Because $\textbf{\textit{u}}^{*}(\lambda^{(2)}_{1},\lambda_{2})$ is known, the solution $\boldsymbol{\beta}^{*}(\lambda^{(3)}_{1},\lambda_{2})$ and $\textbf{\textit{u}}^{*}(\lambda^{(3)}_{1},\lambda_{2})$ can also be solved on the reduced data set, by setting $\tilde{\lambda}_{1}=\lambda^{(2)}_{1}$ and $\lambda_{1}=\lambda^{(3)}_{1}$ in these theorems. With the same process, the solution of fused Lasso can be solved on the reduced data set for every $\lambda_{1}$ and $\lambda_{2}$. Therefore, the computational process of fused Lasso can be speeded up.
\end{Remark}

\section{Numerical Experiments}\label{sec:exp}
In this section, we evaluate the safe feature identification rule on some simulation data sets and real data sets. \textcolor{black}{In order to do so, we apply the function \textit{fusedLeastR} from the SLEP package (Liu et al. \cite{liu2009slep}) to solve fused Lasso. The core idea of fusedLeastR is the Efficient Fused Lasso Algorithm from Liu et al. \cite{liu2010efficient}. Our screening rule is encoded as a new function and implemented before the function fusedLeastR. For every setting of $\lambda_{1}$ and $\lambda_{2}$, we implement this process to eliminate inactive features and identify the features with same coefficients.} All experiments are performed on the MATLAB R2018b with Intel(R) Core(TM) i5-8250U 1.60 CPU and 8G RAM.

To illustrate the efficiency of our rule, we choose 6 values of $\lambda_{2}$ from $\{10^{-4}, 10^{-3}, 10^{-2}, 10^{-1}, 10^{0}, 10^{1}\}$. For each $\lambda_{2}$, we run SLEP along a sequence of 100 tuning parameters equally spaced on the $\lambda_{1}/\lambda^{max}_{1}(\lambda_{2})$ from 0.01 to 1.  Same as the most screening rule papers (Tibshirani et al. \cite{tibshirani2012strong}, Wang et al. \cite{wang2015lasso},  Ndiaye et al. \cite{ndiaye2017gap}, Xiang et al. \cite{xiang2016screening} and so on), we use two quantities to measure the rule, that are rejection ratio and speedup. The rejection ratio is defined as
$\textcolor{black}{\frac{N_{s}}{N_{f}}}$. Under different tuning parameters, $N_{s}$ denotes the number of discarded features by the rule and $N_{f}$ denotes the  actual number of features that are inactive. This ratio measures the efficiency of our rule.  The speedup is defined as $\frac{T_{f}}{T_{s}}$,
where $T_{s}$ and $T_{f}$ indicate the computational time of SLEP with and without our feature identification rule, respectively. Here, $T_{f}$ is the overall time of  SLEP, under 600 different $\lambda_{1}$ and $\lambda_{2}$. $T_{s}$ is the overall time of  SLEP embedded with our  rule, under 600 different $\lambda_{1}$ and $\lambda_{2}$. This quantity, as the name illustrated, indicates  the reduced computational time because of the rule.  The larger the rejection ratio and speedup are, the more efficient the rule is.
\subsection{Simulation}
\subsection{Simulation data}
In this section, we evaluate the safe feature identification rule on some simulation data, which are generated from the true linear regression
$\textbf{\textit{y}}=X\boldsymbol{\beta}^{*}+\epsilon$
with 
\begin{center}
$n\in\{50,100,150,200\}$ and $p\in\left\{1000, 3000, 5000, 8000, 10000, 12000\right\}$. 
\end{center}
So, we consider 24 different data sets. The true coefficient vector $\boldsymbol{\beta}^{*}$ is set as follows.
$$\beta^{*}_{1}=2, \beta^{*}_{3}=1.5, \beta^{*}_{5}=0.8,\beta^{*}_{8}=1,\beta^{*}_{10}=1.75,\beta^{*}_{13}=0.75,\beta^{*}_{16:50}=0.3,$$
and the other elements of $\boldsymbol{\beta}^{*}$ are set as 0. In the true model, $X$ is generated from the multivariate normal distribution with mean vector $\boldsymbol{\mu}$ and covariance matrix $\Sigma$. Here, we consider two cases in \cite{wang2015fused} that are $\Sigma^{1}=I_{p}$ and $\Sigma^{2}_{ij}=0.5^{|i-j|}$. To show the difference between features, as in \cite{shang2021ell, shang2022safe}, we set $\boldsymbol{\mu}$ as
\begin{center}
$\mu_{3:7}=10$,
$\mu_{70:90}=5$, $\mu_{\lfloor\frac{p}{2}\rfloor:\lfloor\frac{2p}{3}\rfloor}=-2,$
\end{center}
and the other elements are set as 0. In addition, elements of the random error $\epsilon$ are generated from the Gaussian distribution with mean 0 and standard variance 0.1 as in \cite{wang2015fused}.

\begin{figure}[htbp]
\begin{minipage}[t]{0.5\linewidth}
\centering
\includegraphics[width=2.5in]{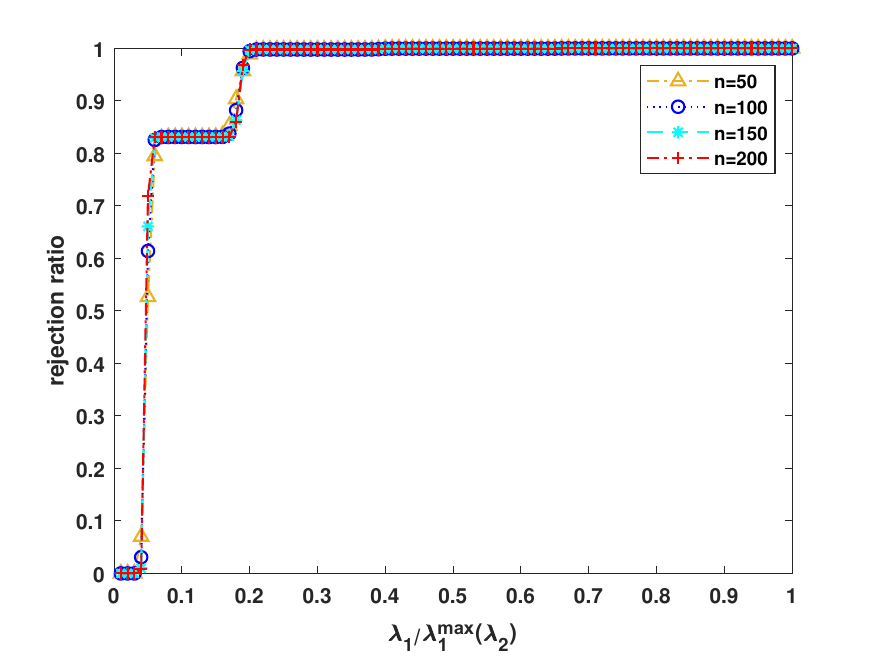}
\end{minipage}
\begin{minipage}[t]{0.5\linewidth}
\centering
\includegraphics[width=2.5in]{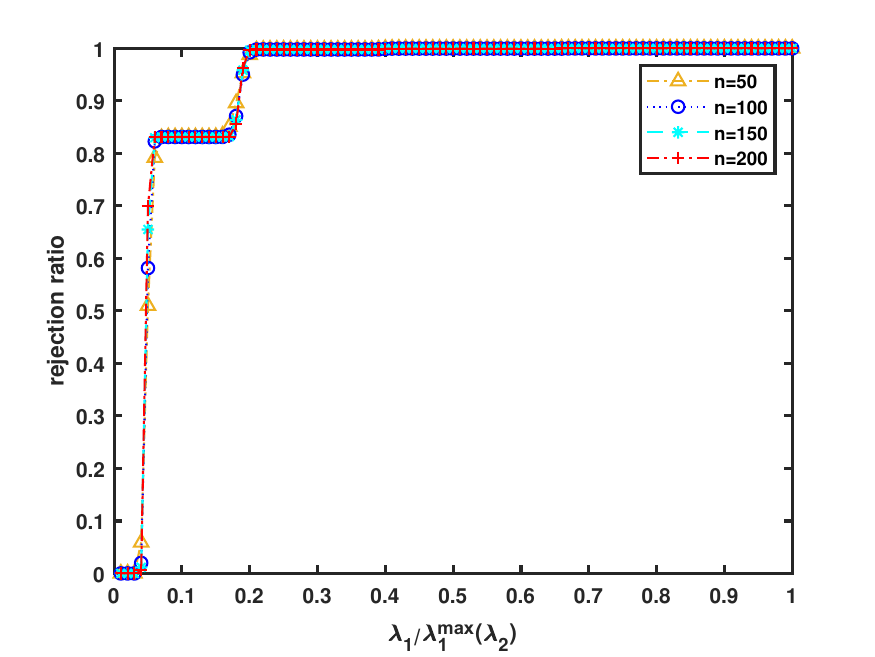}
\end{minipage}
\caption{The rejection ratio on $\Sigma^{1}$ and $\Sigma^{2}$. The left figure is the rejection ratio on $\Sigma^{1}$ and the right one is the result on $\Sigma^{2}$.}
\end{figure}

For every data set, to show the performance of our result, we report the average rejection ratios of 10 simulation in Fig. 1. To make this paper simple, we only show the average rejection ratios of $p=12000$ under different $n$ in Fig. 1. From Figure 1, we know that our screening rule has a good performance on the rejection ratio, no matter the values of $n$. The rejection ratios are all over 0.8 when $\lambda_{1}/\lambda^{max}_{1}(\lambda_{2})$ is less than 0.1 and almost equal to 1 when $\lambda_{1}/\lambda^{max}_{1}(\lambda_{2})$ is larger than 0.1, which means our rule identifies almost all inactive features. From this figure, there are slightly differences between different values $n$ and different covariance matrices. To show the detailed results of our result, we report the speedup in the following tables.

For every data set, to show the performance of our rule, we report the average computational time of 10 simulation. We present the computational time under $\Sigma^{1}$ and $\Sigma^{2}$ in TABLE 1 and TABLE 2, respectively. In each table, we show the average computational time of SLEP, SLEP embedded with our rule and the rule. To be more comprehensive, we also record the standard variance of these computational time.  For instance, the result 0.357(0.002) in TABLE 1 means the average time of our screening rule under 10 simulation is 0.357 and the standard variance is 0.002, when $n=50$ and $p=1000$.
\begin{table}[htbp]
\caption{The computational time of SLEP, SLEP with our rule and our rule, when $\Sigma=I_{p}$. With these time, the speedup can be obtained. Here, SLEP+ means SLEP embedded with our  rule.}
\begin{center}
\begin{tabular}{|c|c|cccc|}
\hline
$n$     & $p$ &SLEP     & \multicolumn{1}{c}{SLEP+} &our rule    & \textbf{speedup}          \\
\hline
\multirow{6}{*}{50}
 & 1000  & 21.18(6.216)    & 3.518(0.032)       &0.357(0.002)   & \textbf{6.581(0.076)}   \\
& 3000  & 70.38(0.132)   &   8.191(0.031)       & 1.156(0.008)  & \textbf{8.592(0.034)}   \\
 & 5000  &122.8(0.275)  &  13.17(0.079)       &1.938(0.006)  &\textbf{9.323(0.061)}\\
 & 8000  &197.0(0.459)  &   20.63(0.050)       &3.306(0.019)   &  \textbf{9.553(0.030)}  \\
& 10000 &262.8(3.119)   &  26.44(0.214)       &4.338(0.024)   &\textbf{9.942(0.093)}    \\
& 12000 &348.1(5.954)   &   33.05(0.180)       &5.299(0.033)  & \textbf{10.53(0.218)}   \\
\hline
\multicolumn{1}{|c|}{\multirow{6}{*}{100}} & 1000  & 20.55(0.149)   & 3.923(0.042)                    & 0.547(0.009)  & \textbf{5.249(0.068)}   \\
\multicolumn{1}{|c|}{} & 3000  &60.25(0.178)    &9.655(0.210)  &1.943(0.013)   & \textbf{6.242(0.131)}   \\
\multicolumn{1}{|c|}{}                                                                                            & 5000  & 114.5(2.300) & 16.23(0.198)                   & 3.675(0.121)  & \textbf{7.037(0.155)}   \\
\multicolumn{1}{|c|}{} & 8000  &253.2(4.349)   &28.97(0.280)       & 5.830(0.009)  & \textbf{8.737(0.566)}   \\
\multicolumn{1}{|c|}{}                                                                                            & 10000 & 350.9(5.127) & 38.80(0.327)                   & 7.241(0.052)  & \textbf{9.044(0.115)}  \\
\multicolumn{1}{|c|}{}                                                                                            & 12000 &500.1(9.795) &51.41(0.551)                   &8.967(0.035)  & \textbf{9.378(0.447)}  \\
\hline
\multirow{6}{*}{150}                                                          & 1000  & 30.36(0.084)   & 4.397(0.030)    & 0.836(0.015)   & \textbf{6.905(0.040)}   \\
& 3000  & 94.42(0.437)   & 11.08(0.086)    & 2.859(0.021)  & \textbf{8.522(0.069)}   \\
& 5000  & 235.3(8.832)   & 22.08(1.492)    & 5.294(0.405)   & \textbf{10.69(0.759)}  \\
& 8000  & 530.1(1.524)   & 39.68(0.059)    & 8.310(0.015)  & \textbf{13.36(0.034)}  \\
& 10000 & 770.2(24.88)   & 53.44(0.229)    & 10.30(0.026)   & \textbf{14.41(0.413)}   \\
& 12000 & 984.0(22.11)   & 67.26(0.702)   & 12.47(0.082)   & \textbf{14.63(0.196)}   \\
\hline
\multicolumn{1}{|l|}{\multirow{6}{*}{200}}
& 1000  &21.64(0.226)    &4.655(0.021)  &1.056(0.008)   &\textbf{4.650(0.062)}    \\
& 3000  &80.16(3.368)   &13.01(0.078)  & 3.930(0.024)  &\textbf{6.159(0.256)}    \\                                                                                            & 5000  &238.9(4.467)    &27.55(0.197)   &6.704(0.015)   &\textbf{8.673(0.146)}  \\
& 8000  &513.9(11.47)    &51.87(0.730)   &10.85(0.035)   &\textbf{9.907(0.189)}    \\
& 10000 &696.9(5.282)    &67.84(0.436)   &13.40(0.058)   &\textbf{10.27(0.118)}  \\                                                                                            & 12000 &884.4(16.62)  &86.80(0.415)   &16.17(0.062)   &\textbf{10.19(0.171)}   \\\hline
\end{tabular}
\end{center}
\end{table}
\begin{table}[htbp]
\caption{The computational time of SLEP, SLEP with our rule and our rule, when $\Sigma_{i,j}=0.5^{|i-j|}$. With these time, the speedup can be obtained. Here, SLEP+ means SLEP embedded with our rule.}
\centering
\begin{tabular}{|c|c|cccc|}

\hline
$n$      & $p$ &SLEP     & \multicolumn{1}{c}{SLEP+} &our rule    & \textbf{speedup}          \\
\hline
\multirow{6}{*}{50}
 & 1000  &21.44(0.165)  &3.811(0.117)        &0.364(0.006)  &\textbf{5.641(0.143)}  \\
& 3000  &64.38(0.390)    &8.802(0.201)        &1.148(0.009)   &\textbf{7.317(0.162)}   \\
 & 5000  &109.4(0.324)   &13.84(0.124)     &1.914(0.007)  &\textbf{7.908(0.086)}    \\
 & 8000  &176.8(1.147)    &21.54(0.217)     &3.244(0.009)   &\textbf{8.210(0.115)}    \\
& 10000 &234.5(3.732) &27.50(0.097)         &4.230(0.018) &\textbf{8.526(0.114)}    \\
& 12000 &307.5(9.097)   &34.48(0.178)       &5.188(0.036)   &\textbf{8.919(0.279)}    \\
\hline
\multicolumn{1}{|c|}{\multirow{6}{*}{100}} & 1000  & 27.43(0.051)   & 3.773(0.950)                    & 0.522(0.004)  & \textbf{6.736(0.071)}   \\
\multicolumn{1}{|c|}{} & 3000  & 83.03(0.493)   &  9.622(0.092)                   &1.914(0.027)   & \textbf{8.617(0.137)}   \\
\multicolumn{1}{|c|}{}                                                                                            & 5000  & 186.5(10.51) & 18.66(0.553)                   & 3.747(0.037)  & \textbf{9.988(0.300)}   \\
\multicolumn{1}{|c|}{} & 8000  & 330.7(6.174)   & 29.23(0.251)      & 5.790(0.011)  & \textbf{11.31(0.288)}   \\
\multicolumn{1}{|c|}{}                                                                                            & 10000 & 539.8(45.45) & 43.28(2.964)                   & 7.619(0.176)  & \textbf{12.46(0.222)}  \\
\multicolumn{1}{|c|}{}                                                                                            & 12000 & 673.2(42.51) & 54.21(5.824)                   & 9.145(0.389)  & \textbf{12.47(0.687)}  \\
\hline
\multirow{6}{*}{150}
& 1000  &31.84(0.281)  &4.776(0.112)  &0.844(0.012) &\textbf{6.670(0.138)}   \\
& 3000  &98.51(0.907)  &11.66(0.228)  &2.905(0.078)  &\textbf{8.451(0.197)}   \\
& 5000  &251.6(18.42)   &21.90(0.363)   &5.095(0.023)   &\textbf{11.50(1.020)}  \\
& 8000  &534.9(4.395)  &41.08(0.286)    &8.360(0.018)   &\textbf{13.02(0.141)}   \\
& 10000 &750.1(1.873)  &56.31(0.981)     &10.46(0.051)   &\textbf{13.32(0.242)}    \\
& 12000 &967.2(5.543)    &68.60(0.894)     &12.58(0.084)    &\textbf{14.10(0.103)}    \\
\hline
\multicolumn{1}{|c|}{\multirow{6}{*}{200}}
& 1000  &21.91(0.131)    &4.851(0.068)   &1.096(0.016)   &\textbf{4.518(0.042)}    \\
& 3000  &82.10(4.288)    &13.40(0.078)   &3.936(0.013)   &\textbf{6.123(0.198)}    \\                                                                                            & 5000  &246.6(9.684)    &27.81(0.200)   &6.674(0.012)   &\textbf{8.865(0.083)}   \\
& 8000  &505.4(8.068)    &52.98(0.709)   &10.86(0.040)   &\textbf{9.538(0.152)}   \\
& 10000 &688.3(11.98)    &69.08(1.087)   &13.05(0.925)   &\textbf{9.966(0.244)}   \\                                                                                            & 12000 &887.8(17.40)    &86.62(1.209)   &17.17(1.450)   &\textbf{10.24(0.132)}   \\\hline
\end{tabular}
\end{table}

From these tables, we can conclude the following statements. (i) Comparing to the computation of the solver SLEP and SLEP embedded with our rule, the average computational time of our rule is relatively small,  which means our rule has a low computational cost. (ii) The standard variance of the computational time of our rule are very small, which means the screening rule is stable. (iii) The value of speed up varies with different data sets. In these tables, the smallest speedup is 4.518 and the largest is 14.63. (iv) For any fixed $n$, the value of speedup increases with the feature size $p$ increasing. For instance, the speedup increases from 8.522 to 14.63 with $p$ increasing from 1000 to 12000, when the sample size $n=150$. (v) For any fixed $p$, the speedup has a decreased tendency when $n$ is larger than 150. When $n=200$, the speedup values are smaller than the corresponding results of 50, 100 and 150. This is because that our rule needs to calculate the values of $|X^{\top}_{.j}\textbf{c}|$ and $\|X_{.j}\|_{2}$, which needs a time complexity $O(n)$. When $n$ is larger than 150, the computational time of these two values increase, which leads to the computational time of the screening rule increased and the speedup decreased. (vi) The speedup values of $\Sigma^{1}$ and $\Sigma^{2}$ have slightly differences, which means our rule fits data sets with different covariance matrices.
\subsection{Real data}
In this section, we evaluate the safe feature identification rule on some real data sets: Prostate \cite{efron2012large}, Srbct \cite{khan2001classification},  20NewsHome, Reuters21578, TDT2 \footnote{http://www.cad.zju.edu.cn/home/dengcai/Data/TextData.html}   and Leukemia \footnote{http://portals.broadinstitute.org/cgi-bin/cancer/datasets.cgi}. To be simple, we omit the introduction of these data sets and just show the size of them in Fig. 2. For these data set, we show rejection ratios in Fig. 2 and the speedup in TABLE 3. \textcolor{black}{In TABLE 3, we add the sparse level under different data sets. Because $N_{f}$ varies with different $\lambda_{1}$ and $\lambda_{2}$, the index sparse level is defined as the interval of $N_{f}/p$ when $\lambda_{1}/\lambda^{max}_{1}(\lambda_{2})\in[0.01:0.01:1]$. This index reflects sparse level of the model solution under $\lambda_{1}$ and $\lambda_{2}$.}
\begin{figure}[htbp]
\subfigure[Prostate $(102\times6033)$]{
\begin{minipage}[t]{0.5\linewidth}
\centering
\includegraphics[width=2.5in]{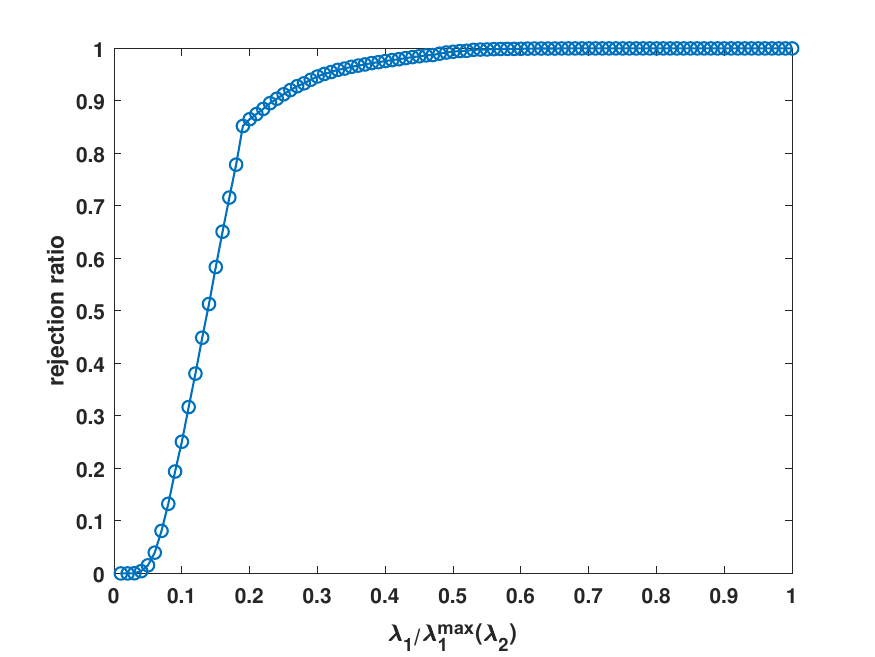}
%\caption{fig2}
\end{minipage}
}
\subfigure[Srbct $(63\times2308)$]{
\begin{minipage}[t]{0.5\linewidth}
\centering
\includegraphics[width=2.5in]{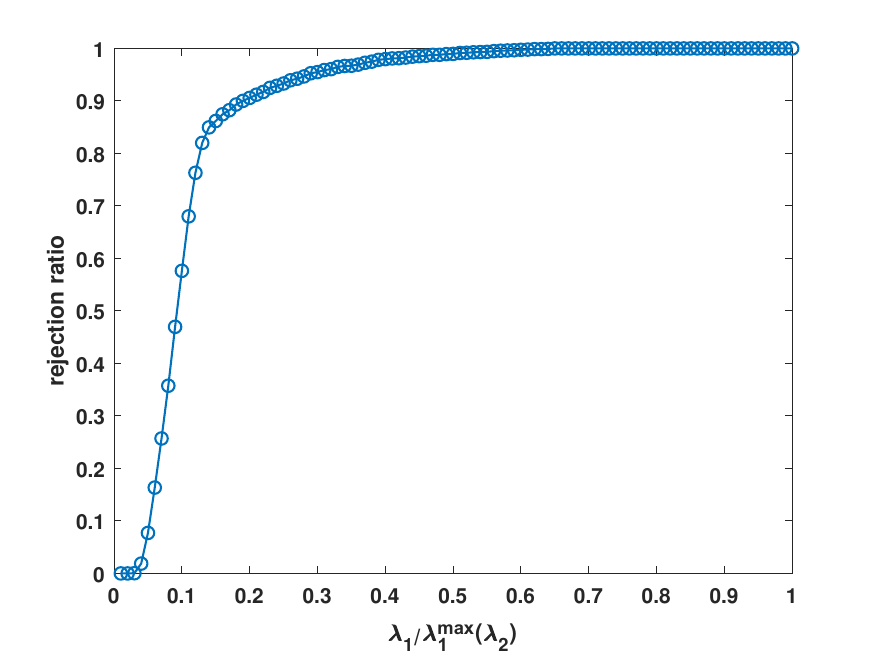}
%\caption{fig2}
\end{minipage}
}

\subfigure[Reuters21578 $(8293\times18933)$]{
\begin{minipage}[t]{0.5\linewidth}
\centering
\includegraphics[width=2.5in]{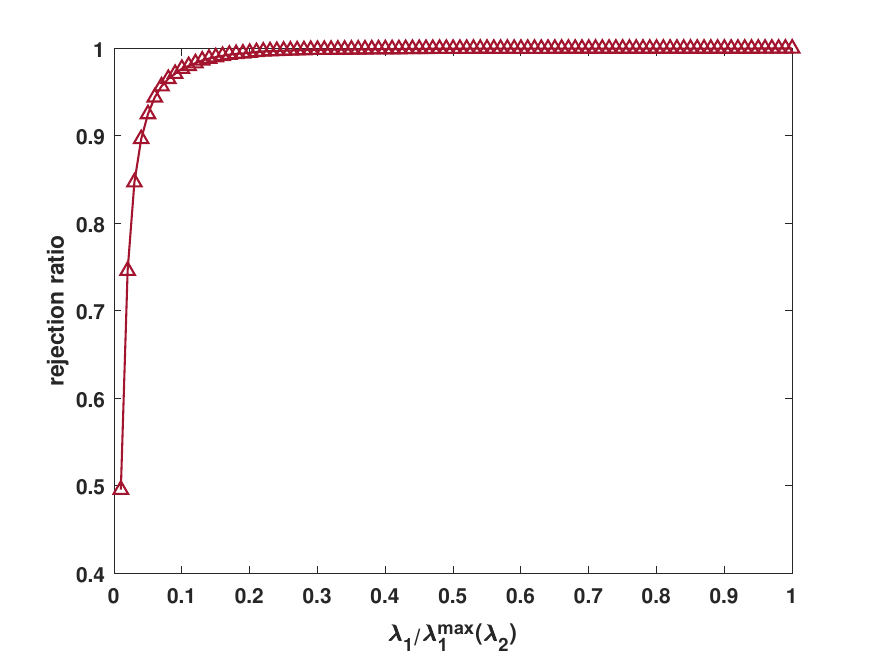}
%\caption{fig2}
\end{minipage}
}%
\subfigure[20NewsHome $(18774\times61188)$]{
\begin{minipage}[t]{0.5\linewidth}
\centering
\includegraphics[width=2.5in]{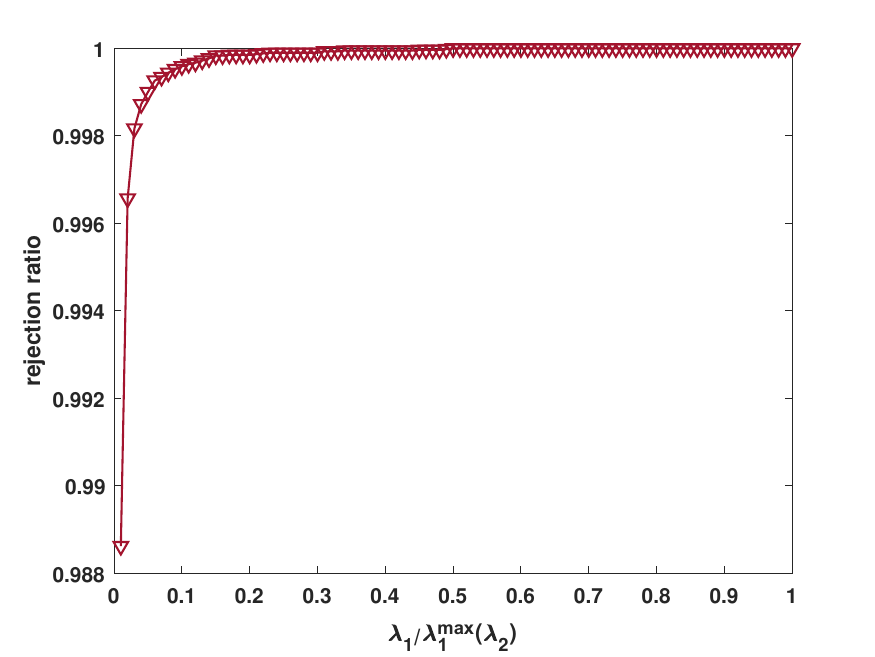}
%\caption{fig2}
\end{minipage}
}

\subfigure[TDT2 $(500\times36771)$]{
\begin{minipage}[t]{0.5\linewidth}
\centering
\includegraphics[width=2.5in]{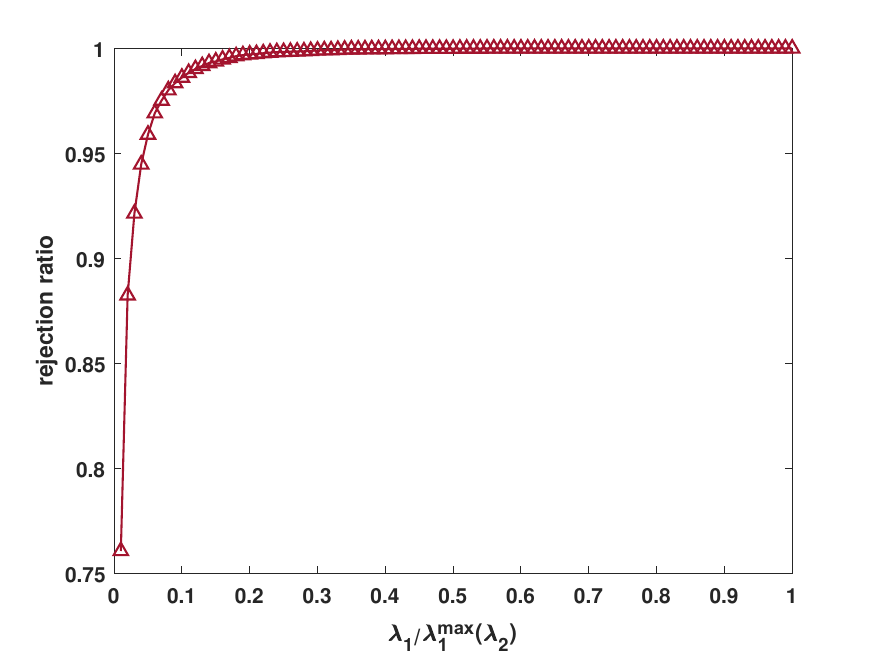}
%\caption{fig2}
\end{minipage}
}
\subfigure[Leukemia $(38\times7128)$]{
\begin{minipage}[t]{0.5\linewidth}
\centering
\includegraphics[width=2.5in]{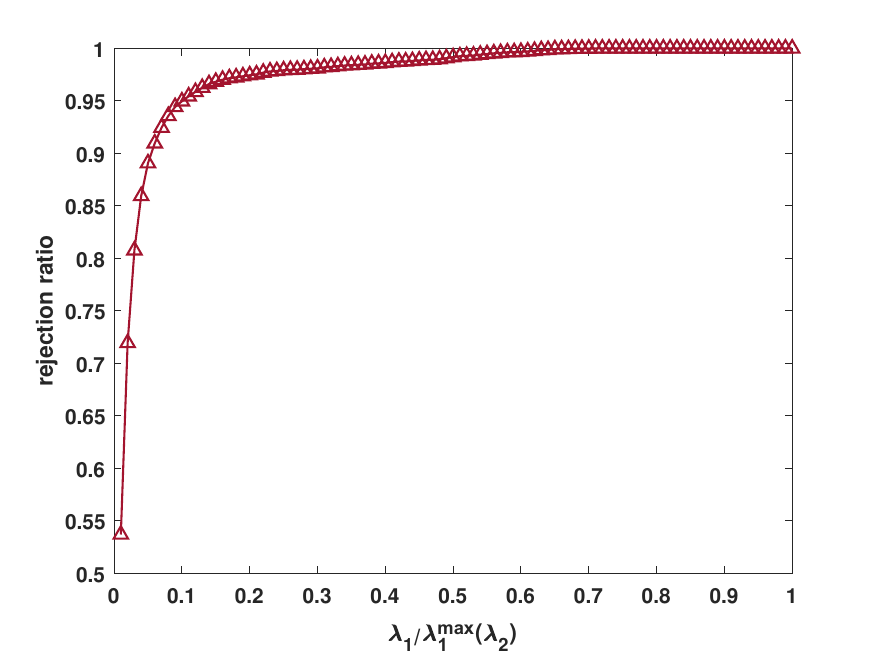}
%\caption{fig2}
\end{minipage}
}
\caption{The rejection ratio under different data sets.}
\end{figure}
\begin{table}[htbp]
\caption{The speedup of different real data sets. Here, SLEP+ means SLEP embedded with our rule.}
\centering
\begin{tabular}{|c|ccccc|}
\hline
data sets     & SLEP   & SLEP+ & our rule & \textbf{speedup} &\textcolor{black}{sparse level}\\ \hline
Prostate      & 177.6 & 38.72               & 4.460         & \textbf{4.586}  &\textcolor{black}{[0.998,1]} \\
Srbct         & 49.10  & 10.14               & 1.037         & \textbf{4.843} &\textcolor{black}{[0.990,1]} \\
20NewsHome    & 706.0 & 143.6                & 97.14        & \textbf{4.916}  &\textcolor{black}{[0.999,1]}\\
Reusters21578 & 492.5 & 67.70                & 24.99        & \textbf{7.274}  &\textcolor{black}{[0.996,1]} \\
TDT2          & 425.0 & 52.00                & 37.47        & \textbf{8.174}  &\textcolor{black}{[0.998,1]}\\
Leukemia      & 155.0 & 8.781                & 2.330         & \textbf{17.65} &\textcolor{black}{[0.998.1]}\\ \hline
\end{tabular}
\end{table}

According to Fig. 2 and TABLE 3, we know that our rule has a good performance on these real data sets. From results in Fig. 2, we know that our rule has a good performance in rejection ratio, especially on Reuters21578, 20NewsHome, TDT2 and Leukemia. On these four data sets, the rejection ratios of $\lambda_{1}/\lambda^{max}_{1}(\lambda_{2})=0.01$ are larger than 0.5, 0.98, 0.75 and 0.5, respectively. In addition, the rejection ratio of these data sets reach to 1 rapidly fast. Our rule has a better performance on rejection ratios of Reuters21578, 20NewsHome and TDT2 than Leukemia, the speedup of Reuters21578, 20NewsHome and TDT2 should be larger than that of Leukemia, because that our rule eliminates almost all inactive features under every $\lambda_{1}$ and $\lambda_{2}$ for these data sets. However, in TABLE 3, we find that the speedup of Leukemia is the largest one. This is because that our rule has to calculate the values of $|X^{\top}_{.j}\textbf{c}|$ and $\|X_{.j}\|_{2}$, which needs a time complexity $O(n)$. Although the rejection ratios of Reuters21578, 20NewsHome and TDT2 are very ideal, their sample sizes are relatively large, which makes the speedup decreased.

\section{Conclusion}\label{sec:con}
To accelerate the computation of fused Lasso, we propose a safe feature identification rule with the help of an extra dual variable and variational inequality. It is easy to see that the dual problem of fused Lasso contains two variables, while its objective function is only related with one variable. Therefore, we estimate one dual solution via the variational inequality and calculate the bound of the other dual solution under the former estimation. With these estimations, we obtain the safe feature identification rule which has a low computational cost to eliminate inactive features and identify adjacent features with same coefficients.  In the numerical experiments, we evaluate our rule on simulation and real data, which illustrate that our rule is efficient on speeding up the computation of fused Lasso. To the best of our knowledge, the existing screening rules can not deal with fused Lasso, because that there are two regularizers and the conjugate function of the fused term does not have a simple form.
\section*{Acknowledgements}
This work was supported by the National Natural Science Foundation of China (12071022, 12371322), the Project funded by China Postdoctoral Science Foundation (2022M723327) and the National Natural Science Foundation of Beijing (Z220001).
\bibliographystyle{plain}
\bibliography{reference}

% insert where needed to balance the two columns on the last page with
% biographies
%\newpage

% You can push biographies down or up by placing
% a \vfill before or after them. The appropriate
% use of \vfill depends on what kind of text is
% on the last page and whether or not the columns
% are being equalized.

%\vfill

% Can be used to pull up biographies so that the bottom of the last one
% is flush with the other column.
%\enlargethispage{-5in}

% that's all folks
\end{document}